\documentclass[a4paper,11pt,reqno]{amsart}

\usepackage[margin=1in,includehead,includefoot]{geometry}
\usepackage[utf8]{inputenc}
\usepackage[T1]{fontenc}
\usepackage{lmodern,microtype}
\usepackage{hyperref}
\usepackage{booktabs}
\usepackage{mathtools,amssymb,amsthm,amsmath}
\usepackage[foot]{amsaddr}
\usepackage{mathrsfs}
\usepackage{graphicx}
\usepackage[table]{xcolor}
\usepackage{wrapfig}
\usepackage[margin=5mm]{caption}
\usepackage{enumitem}
\usepackage[textsize=footnotesize,disable]{todonotes}
\usepackage[notref,notcite,final]{showkeys}
\usepackage{mdframed}
\usepackage{xpatch}
\usepackage{pgfplots}
\usepackage{mycommands}

\graphicspath{{./graphics/}}

\makeatletter
\let\old@setaddresses\@setaddresses
\def\@setaddresses{\bigskip{\parindent 0pt\let\scshape\relax\let\ttfamily\relax\old@setaddresses}}
\makeatother

\newtheorem{theorem}{Theorem}

\newtheorem{lemma}[theorem]{Lemma}
\theoremstyle{remark}

\hypersetup{
  pdftitle={A simple algorithm for computing Hamilton paths on independent set polytopes},
  pdfauthor={Jean Cardinal, Pia Herkenrath, Torsten M\"utze and Francesco Verciani}
}

\begin{document}

\title{A simple algorithm for computing Hamilton paths on independent set polytopes}

\author{Jean Cardinal}
\address[Jean Cardinal]{Computer Science Department, Universit\'e Libre de Bruxelles, Belgium}
\email{jean.cardinal@ulb.be}

\author{Pia Herkenrath}
\address[Pia Herkenrath]{Institut f\"ur Mathematik, Universit\"at Kassel, Germany}
\email{pherkenrath@mathematik.uni-kassel.de}

\author{Torsten M\"utze}
\address[Torsten M\"utze]{Institut f\"ur Mathematik, Universit\"t Kassel, Germany}
\email{tmuetze@mathematik.uni-kassel.de}

\author{Francesco Verciani}
\address[Francesco Verciani]{Institut f\"ur Mathematik, Universit\"at Kassel, Germany}
\email{francesco.verciani@mathematik.uni-kassel.de}

\thanks{This project was supported by German Science Foundation grant~522790373.}

\maketitle
\begin{abstract}
The independent set polytope, or stable set polytope, of a graph~$G$ is the 0/1-polytope defined by the convex hull of the characteristic vectors of all independent sets of~$G$.
We present a simple algorithm for computing a Hamilton path on the independent set polytope of a given $n$-vertex graph~$G$ with amortized delay~$\cO(n)$.
The independent sets are listed such that two consecutive sets differ either in removing a vertex, or adding a vertex and removing its neighbors from the independent set, i.e., the symmetric difference between two consecutive independent sets induces a star in~$G$.
As applications of this result, we obtain an algorithm to compute a Hamilton path on the matching polytope of an $m$-edge graph~$G$ with worst-case delay~$\cO(m)$, which lists all matchings of~$G$ in such a way that the symmetric difference between two consecutive matchings is a path on at most three edges.
Furthermore, we obtain an algorithm to compute a Hamilton path on the chain polytope and order polytope of an $n$-element poset~$P$ with amortized delay~$\cO(n)$, which lists all antichains of~$P$ or all ideals of~$P$, respectively, by star exchanges.
Our algorithms are derived from the generic framework proposed by Merino and M\"utze [Traversing combinatorial 0/1-polytopes via optimization. SIAM J. Comput., 53(5):1257--1292, 2024] for computing Hamilton paths on arbitrary 0/1-polytopes, which uses a linear optimization procedure as a black box.
Our algorithms bypass solving the computationally intractable maximum weight independent set problem by a simple and purely combinatorial greedy rule.
\end{abstract}

\section{Introduction}

Polytopes with vertex coordinates in $\{0,1\}$, called \defi{0/1-polytopes}, are ubiquitous in combinatorial optimization.
This is because the solutions of many interesting optimization problems can be described as 0/1-vectors, by considering the characteristic vectors of the underlying sets.
For example, consider a graph~$G$ with $m$ edges labeled $1,\ldots,m$, and let~$M$ be a matching in~$G$, i.e., a set of edges, no two of which have a vertex in common.
Then the \defi{characteristic vector} of~$M$, denoted~$x^M\in\{0,1\}^m$ is defined as $x^M_e=1$ if $e\in M$ and $x^M_e=0$ if $e\notin M$.
The \defi{matching polytope} of~$G$ is defined as~$\scM(G)\coloneq\conv(\{x^M\mid M \text{ matching in }G\})\subset\mathbb{R}^m$ \cite{MR371732}.
The maximum weight matching problem is to maximize the linear objective function $w\cdot x$ for all $x\in\scM(G)$, where $w\in \mathbb{R}^m$ is the vector of weights assigned to the edges of~$G$.
This problem has been a cornerstone of computer science research in the past decades, and can be solved efficiently with one of the many algorithms that are now taught in textbooks.

\begin{figure}[t!]
\includegraphics[page=2]{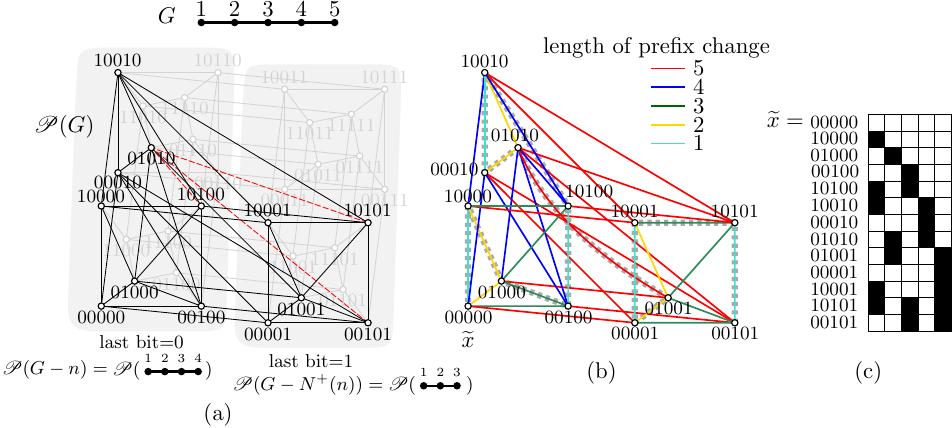}
\caption{(a) Independent set polytope of the 5-vertex path, with vertices numbered~$1,\ldots,5$ from left to right.
Characteristic vectors of independent sets of the path are bitstrings of length~5 that avoid two adjacent 1s.
Edges of the polytope corresponding to star exchanges are solid, the others dashed.
The shading highlights the two faces of the polytope in which the last bit is either~0 or~1.
(b)~Hamilton path computed by Algorithm~I on the independent set polytope from the left, with initial vertex~$\tx=00000$.
The edges are colored according to the length of the prefix change, and the algorithm changes the shortest possible prefix that leads to a new vertex.
(c)~Vertices listed in the order of the Hamilton path, with 0-bits and 1-bits visualized as white and black squares, respectively, revealing the genlex property.
}
\label{fig:ind}
\end{figure}

In this work, we focus on the independent set polytope of a graph~$G$, also known as stable set polytope~\cite{MR371732}.
Specifically, an \defi{independent set} in a graph~$G$ is a set of vertices that spans no edges, and the \defi{independent set polytope} is defined as $\scP(G)\coloneq \conv(\{x^U\mid U\text{ independent set in }G\})$; see Figure~\ref{fig:ind}~(a).
We note that matchings of~$G$ can be seen as independent sets in the \defi{line graph}~$L(G)$, the graph that has a vertex for every edge of~$G$, and an edge between vertices corresponding to edges of~$G$ that are incident to each other.
Consequently, the matching polytope of~$G$ equals the independent set polytope of~$L(G)$.
For example, the independent set polytope of the 5-vertex path shown in Figure~\ref{fig:ind} equals the matching polytope of the 6-vertex (=5-edge) path.
Unlike the aforementioned maximum weight matching problem, which is polynomial-time solvable, the corresponding optimization problem on the independent set polytope is NP-hard, as it includes the maximum cardinality independent set problem on~$G$ as a special case (all weights equal to~1).

Combinatorial 0/1-polytopes as the ones discussed before encode interesting information about the underlying combinatorial objects.
For example, the edges of the matching polytope connect pairs of matchings whose symmetric difference is an alternating path or cycle.
This is a consequence of the fact that edges of the independent set polytope connect pairs of independent sets whose symmetric difference induces a connected subgraph~\cite{MR371732}.
The structure of these polytopes has also been studied extensively with the goal of understanding the aforementioned apparent dichotomy between tractability and intractability of the corresponding optimization problems.

\subsection{Hamilton paths on 0/1-polytopes}

In this work, we consider the problem of computing a Hamilton path on the independent set polytope of a given graph~$G$.
The goal is to visit every vertex of the polytope exactly once, such that any two consecutive vertices are connected by an edge.
One crucial property of such an algorithm is its \defi{delay}, which is the time between any two consecutive vertices of the Hamilton path and before the first and after the last vertex.
The delay can be specified as a worst-case guarantee, which is the maximum of the aforementioned quantities, or as an amortized guarantee, which is the average.
The amortized delay is simply the total time of the algorithm to compute the entire path, divided by the total number of vertices.

Combinatorially, we aim to compute every independent set of~$G$ exactly once, such that the symmetric difference between any two consecutive independent sets induces a connected subgraph of~$G$.
We note that listings of combinatorial objects subject to some minimum-change condition are also known as \defi{Gray codes} in the literature~\cite{MR1491049,MR4649606}.

By a classical result of Naddef and Pulleyblank~\cite{MR762893}, every 0/1-polytope admits a Hamilton path.
Recently, Merino and M\"utze developed a generic algorithm for computing a Hamilton path on any 0/1-polytope~\cite{MR4795009}, which was subsequently improved by Fink, Hlad\'ik, Merino, Mi{\v{c}}ka and~M\"utze~\cite{fink-et-al:2026}.
Their algorithm uses as a black box an algorithm for solving the linear optimization problem over the corresponding polytope, and the resulting delay is given by the time to solve the linear optimization problem.
As mentioned before, linear optimization over the independent set polytope is hard, so this approach does not directly give an efficient algorithm for computing a Hamilton path on the independent set polytope.

\subsection{Our results}
\label{sec:results}

Our main contribution is a simple algorithm for computing a Hamilton path on the independent set polytope~$\scP(G)$ of a given graph~$G$; see Figure~\ref{fig:ind}~(b)+(c).
The algorithm comes in a basic variant and an improved variant for faster speed (see Algorithms~I and~I\sss{} and Theorems~\ref{thm:algo-basic} and~\ref{thm:algo-faster} below).
The algorithm has a number of desirable features:
\begin{itemize}[itemsep=0ex,parsep=0ex,leftmargin=0ex,itemindent=0ex,labelsep=1ex,labelwidth=-2ex]
\item It works for every graph~$G$, and for every labeling of the vertex set of~$G$ with integers~$1,\ldots,n$, which gives additional freedom that can be exploited in some applications.
\item It works for every starting vertex of the polytope, i.e., for every initial independent set.
\item The algorithm is purely combinatorial: it does not perform any polyhedral or numerical computations, and also does not rely on solving the corresponding (hard) optimization problem, and it can be implemented with few lines of code.
\item The resulting amortized delay is linear in the number of vertices of~$G$.
\item To move from one independent set to the next, the algorithm applies one of two operations, referred to as a \defi{star exchange}: either remove a vertex of~$G$, or add a vertex of~$G$ and remove all its neighbors from the current independent set.
Consequently, the subgraph of~$G$ induced by the symmetric difference between two consecutive independent sets is always a star.
This also yields an algorithmic proof that the subgraph of the independent set polytope corresponding to star exchanges admits a Hamilton path.
\item The computed ordering of independent sets has the so-called genlex property, which means that characteristic vectors with the same suffix appear consecutively.
\end{itemize}

We also present the following applications of our algorithm.

\begin{itemize}[itemsep=0ex,parsep=0ex,leftmargin=0ex,itemindent=0ex,labelsep=1ex,labelwidth=-2ex]
\item
By applying this algorithm to the $n$-vertex graph without edges, we recover the well-known binary reflected Gray code (Section~\ref{sec:families}).
By applying the algorithm to paths and cycles, we obtain Gray codes for Fibonacci and Lucas words, respectively (Theorem~\ref{thm:fib-luc}).

\item
As vertex covers of a graph~$G$ are complements of independent sets of~$G$, we also obtain an algorithm for computing a Hamilton path on the vertex cover polytope of~$G$ (Section~\ref{sec:cover}).

\item
As mentioned before, matchings in a graph~$G$ are independent sets in its line graph~$L(G)$.
We thus obtain an algorithm for computing a Hamilton path on the matching polytope~$\scM(G)$, where any two consecutive matchings differ in an alternating path of length at most~3 (see Algorithm~M and Theorem~\ref{thm:matching}).
The (worst-case) delay is~$\cO(m)$, where $m$ is the number of edges of~$G$.

\item
Antichains in a poset~$P$ are independent sets in the comparability graph of~$P$.
We thus obtain an algorithm for computing a star exchange Hamilton path on Stanley's~\cite{MR824105} chain polytope~$\scC(P)$ (see Algorithm~A\sssm{} and Theorem~\ref{thm:antichain}).
In every step, either an element is removed from the current antichain, or an element is added and all elements in its downset are removed.
The amortized delay per visited antichain is~$\cO(n)$, where $n$ is the number of elements of~$P$.
Similarly, we obtain an algorithm for computing a star exchange Hamilton path on Stanley's order polytope~$\scO(P)$ (Theorem~\ref{thm:ideal}), with amortized delay~$\cO(n)$ per visited ideal of~$P$.
In every step, either a maximal element from the ideal is removed together with all elements that are uniquely covered by it, or an element is added together with its downset.
\end{itemize}

\subsection{Related work}

Kurita, Wasa, Uno and Arimura~\cite{MR4262107} describe an algorithm for listing all independent sets of a graph~$G$ with amortized delay~$\cO(\omega(G))$, where $\omega(G)$ is the clique number of~$G$.
In particular, for graphs with bounded clique number, the delay is constant.
The computed listing is not a Hamilton path on the independent set polytope, though.

As our algorithm lists all independent sets of a given graph, one may wonder naturally about listing only the maximum independent sets, or only the inclusion-maximal ones.
In fact, the maximum independent set polytope is a face of the independent set polytope, obtained by intersecting it with a hyperplane that fixes the cardinality of the sets.
As the maximum independent set problem is NP-hard in general graphs, there is likely no polynomial delay algorithm for listing maximum independent sets.
Tsukiyama, Ide, Ariyoshi and Shirakawa~\cite{MR476582} give an algorithm for listing all inclusion-maximal independent sets with delay~$\cO(mn)$ (see also~\cite{MR774940}), where here and in the following, $m$ and $n$ denote the number of edges and vertices of an input graph, respectively.
This was later improved by Makino and Uno~\cite{MR2159537}.
Johnson, Yannakakis and Papadimitriou~\cite{MR933271} provided an algorithm for listing all inclusion-maximal independent sets in lexicographic order with delay~$\cO(mn)$.

Kashiwabara, Masuda, Nakajima and Fujisawa~\cite{MR1146339} list maximum independent sets of a \emph{bipartite} graph~$G$ with delay~$\cO(1)$, but not according to the associated polytope.
Using the aforementioned reduction to an optimization problem, Merino and M\"utze~\cite{MR4795009} obtained an algorithm for computing a Hamilton path on the independent set polytope and maximum independent set polytope of a bipartite graph, with delay~$\cO(m\sqrt{n})$ and~$\cO(mn)$, respectively.\footnote{The delays stated here are without the $\log n$ factor that can be saved as described in~\cite{fink-et-al:2026}.}
Recall that the algorithm for the independent set polytope presented here has only amortized delay~$\cO(n)$ and works for general graphs.

Merino and M\"utze~\cite{MR4795009} also gave an algorithm for computing Hamilton paths on the matching polytope~$\scM(G)$ and perfect matching polytope of any graph~$G$, with delay $\cO(m\sqrt{n})$ and~$\cO(m\sqrt{n}\log n)$, respectively.
For the matching polytope, our new algorithm with delay~$\cO(m)$ improves on this.
For comparison, the algorithm of Uno~\cite{MR1917757} lists perfect matchings of a \emph{bipartite} graph with delay~$\cO(\log n)$, but not according to the polytope.

The authors of~\cite{MR4795009} also obtained an algorithm for computing Hamilton paths on the chain polytope~$\scC(P)$ and order polytope~$\scO(P)$ of any poset~$P$, with delay $\cO(n^{2.5}/\sqrt{\log n})$ and $\cO(n^4)$, respectively, where $n$ is the number of elements of~$P$.
The amortized delay~$\cO(n)$ reported in this paper is a considerable improvement.

In the following, we describe a number of results on listing ideals of a poset in $k$-Gray code order, which means that the symmetric difference of any two consecutive ideals has size at most~$k$.
Note that a 1-Gray code is necessarily a Hamilton path on the order polytope, but a $k$-Gray code for $k\geq 2$ is in general not.
By considering the maximal elements of the ideals, a 1-Gray code for ideals also gives a Hamilton path on the chain polytope~$\scC(P)$ (see Section~\ref{sec:anti-ideal} below for details).

Koda and Ruskey~\cite{MR1231447} proved that ideals of a poset whose Hasse diagram is a forest of trees rooted at a minimal element admits a 1-Gray code that can be computed with delay~$\cO(1)$.

Pruesse and Ruskey~\cite{MR1267190} proved that the ideals of any poset admit a 2-Gray code, computable with delay~$\cO(n)$.
Habib, Nourine, Steiner~\cite{MR1474594} present a 2-Gray code for interval orders that can be computed with delay~$\cO(1)$.
For general posets~$P$, Habib, Medina, Nourine, Steiner~\cite{MR1828422} presented another 2-Gray code, computable with delay that is linear in the maximum down-degree of the Hasse diagram of~$P$, but each ideal appears twice in the listing.
In a recent breakthrough, Brenner and Fink~\cite{brenner-fink:2026} presented an algorithm for listing the ideals and antichains of any poset with amortized delay~$\cO(1)$, even in 3-Gray code order.
As mentioned before, these Gray codes are in general not Hamilton paths on the associated polytopes.

\subsection{Outline}

In Section~\ref{sec:algo} we present our basic algorithm for computing a Hamilton path on the independent set polytope, and an improved variant of the algorithm for higher speed.
In Section~\ref{sec:appl} we discuss in detail the aforementioned applications of our algorithm.
At the end of this paper in Section~\ref{sec:exp}, we present an extensive experimental comparison of both variants of our algorithm on a variety of input graphs.

\section{The algorithm}
\label{sec:algo}

\subsection{Preliminaries}

We define $[n]\coloneq \{1,\ldots,n\}$.
Throughout this paper we consider graphs~$G=([n],E)$ that have $[n]$ as their vertex set.
This corresponds to labeling the vertices of~$G$ with integers~$1,\ldots,n$.
The algorithms discussed later depend on this labeling, hence changing the labeling of the graph yields a different output.
Given a graph~$G=([n],E)$ we write $\rev(G)$ for the graph in which the labeling is reversed: label~$i$ is replaced by $n-i+1$ for all~$i\in[n]$.
For a subset $U\seq [n]$ we write $G[U]$ for the subgraph of~$G$ induced by~$U$.
We also write $G-U$ for the graph obtained from~$G$ by deleting all vertices in~$U$ and their incident edges.
For a singleton~$U=\{i\}$ we simply write $G-i\coloneq G-\{i\}$.
For a vertex $i\in [n]$ we write $N(i)\coloneq \{j\in [n]\mid ij\in E\}$ and $N^+(i)\coloneq N(i)\cup\{i\}$ for the \defi{open} and \defi{closed neighborhood} of~$i$ in~$G$, respectively.
We also partition the open neighborhood~$N(i)$ into two sets, namely those that are smaller or larger than~$i$, denoted by $N_<(i)\coloneq \{j\in N(i)\mid j<i\}$ and $N_>(i)\coloneq \{j\in N(i)\mid j>i\}$, respectively.

For a graph~$G$ with vertex set~$[n]$ we define the \defi{backward and forward degeneracy} as $d_<(G)\coloneq \max\{|N_<(i)|\mid i\in[n]\}$ and $d_>(G)\coloneq \max\{|N_>(i)|\mid i\in [n]\}$, respectively.
Note that we have $d_<(G)=d_>(\rev(G))$ and $d_<(\rev(G))=d_>(G)$, i.e., they only differ in reversing the order of vertices of~$G$.
In the literature there is only one notion of degeneracy, but in the runtime analysis of our algorithms both variants appear simultaneously, which is why we distinguish them.
If the labeling of vertices of a graph~$G$ is not given, the smallest possible (forward) degeneracy of a labeling of~$G$ and a corresponding (forward) ordering can be computed by repeatedly removing a vertex with minimum degree.
The maximum degree of~$G$ is denoted by~$\Delta(G)$, and we clearly have $d_<(G)\leq \Delta(G)$ and $d_>(G)\leq \Delta(G)$.

Furthermore, we write $\cI(G)$ for the collection of independent sets in~$G$, i.e.,
\[ \cI(G)\coloneq \{U\seq [n]\mid U\text{ independent set in }G\}=\{U\seq [n]\mid G[U]\text{ has no edges }\}. \]
For any graph~$G=([n],E)$ and any vertex~$i\in[n]$, the set~$\cI(G)$ satisfies the recursion
\begin{equation}
\label{eq:IG-rec}
\cI(G)=\cI(G-i)\cup \big\{J\cup \{i\} \mid J\in \cI(G-N^+(i))\big\}.
\end{equation}
The first set on the right hand side contains all independent sets that do not contain the vertex~$i$, which are the same as independent sets in the graph obtained from~$G$ by removing the vertex~$i$, and the second set contains all independent sets that contain the vertex~$i$, which are the same as independent sets in the graph obtained from~$G$ by removing the closed neighborhood of~$i$, plus the vertex~$i$ added.

Given an independent set~$U\in \cI(G)$, we write $x^U\in\{0,1\}^n$ for its \defi{characteristic vector}, with $x^U_i=1$ if $i\in U$ and $x^U_i=0$ if $i\notin U$.
We define $X(G)\coloneq \{x^U\mid U\in \cI(G)\}\seq \{0,1\}^n$ for the set of all characteristic vectors of independent sets of~$G$.
We write $\scP(G)\coloneq \conv(X(G))=\conv(\{x^U\mid U\text{ independent set in }G\})\subset \mathbb{R}^n$ for the \defi{independent set polytope}, also known as stable set polytope, of~$G$.
The recursion~\eqref{eq:IG-rec} corresponds to intersecting the polytope~$\scP(G)$ with two hyperplanes $x_i=0$ and~$x_i=1$, which yields two faces of~$\scP(G)$, namely the independent set polytopes~$\scP(G-i)$ and~$\scP(G-N^+(i))$; see Figure~\ref{fig:ind}~(a).
The following adjacency condition on the independent set polytope is well-known and has already been mentioned in the introduction.
We denote the  \defi{symmetric difference} of two sets $U,V$ by $U\triangle V\coloneq (U\setminus V)\cup (V\setminus U)$.

\begin{lemma}[\cite{MR371732}]
\label{lem:edges-polytope}
For any two distinct independent sets~$U,V\in \cI(G)$, the pair of vertices $\{x^U,x^V\}$ is an edge of the polytope~$\scP(G)$ if and only if $G[U\triangle V]$ is connected.
\end{lemma}

For two independent sets~$U,V\in\cI(G)$, the transition $U\rightarrow V$ is called a \defi{star exchange} if either $V=U\setminus \{k\}$ for some $k\in U$ or $V=(U\cup\{k\})\setminus N(k)$ for some $k\notin U$, hence either a single vertex is removed from~$U$, or a single vertex is added to~$U$ and its neighbors are removed from it.
In both cases, $G[U\triangle V]$ is a star graph.
In particular, such a transition corresponds to an edge of the polytope, but not every edge of the polytope is of this type; see Figure~\ref{fig:ind}~(a).

For a string~$x$ and integer~$k\geq 1$, we write $x^k$ for the $k$-fold repetition of~$x$.
For a bit $b\in\{0,1\}$ we write $\ol{b}$ for its complement.

A listing $L=(x_1,\ldots,x_N)$ of binary strings $x_i\in\{0,1\}^n$ is called \defi{genlex} if all strings with the same suffix appear consecutively in~$L$.
Equivalently, this means that all strings ending with~0 appear before all strings ending with~1, or vice versa, and this property is true recursively within each block obtained by ignoring the same last bit.
A genlex listing of the set~$\cI(G)$ lists all elements of the first set on the right hand side of~\eqref{eq:IG-rec} (for $i=n$) before all elements of the second set, or vice versa.
Thus, a genlex Hamilton path in~$\scP(G)$ is one that visits all vertices of the face~$\scP(G-n)$ before all vertices of the face~$\scP(G-N^+(n))$, flipping the last bit at most once, and this property is true recursively within each face obtained by projecting out the last coordinate; see Figure~\ref{fig:ind}~(b)+(c).

\subsection{Basic algorithm}

Our algorithm for computing a Hamilton path on the independent set polytope is stated in pseudocode as Algorithm~I.
It takes as input a graph~$G=([n],E)$ and an initial independent set~$\tx\in X(G)$.
One possible initialization that is valid for any graph is the empty set, that is, $\tx=0^n$.
After the initialization in step~I1, the algorithm repeatedly performs a loop consisting of steps~I2--I5.
The current independent set is stored in the variable~$x\in\{0,1\}^n$, which is initialized in step~I1, visited in step~I2, and updated in step~I4, before at the end of step~I5 the control loops back to the next visit step.
The \emph{only} additional data structure used by the algorithm is a binary array $b=(b_1,\ldots,b_n)$ that is initialized in step~I1, used in step~I3 to compute the length~$k$ of the prefix of~$x$ that will be modified, and updated in step~I5.

\begin{algo}{Algorithm~I}{Independent sets by star exchanges on shortest prefixes}
For a graph~$G=([n],E)$, this algorithm computes a Hamilton path on the independent set polytope~$\scP(G)$, starting from an initial independent set~$\tx\in X(G)$.
\begin{enumerate}[label={\bfseries I\arabic*.}, leftmargin=8mm, noitemsep, topsep=3pt plus 3pt]
\item{} [Initialize] Set $x \gets \tx$ and $b\leftarrow 1^n$.
\item{} [Visit] Visit~$x$.
\item{} [Shortest prefix change] Compute $k\gets \min\{i\in[n]\mid b_i=1\wedge \sum_{j\in N_>(i)}x_j=0\}$.
Terminate if $k=\infty$.
\item{} [Update $x$] Set $x_k\gets \ol{x_k}$.
If $x_k=1$, then for all $i\in N_<(k)$ set $x_i\gets 0$.
\item{} [Update $b$] Set $b_k\gets 0$, and for $i=1,\ldots,k-1$ set $b_i\gets 1$, then goto~I2.
\end{enumerate}
\end{algo}

The idea of the algorithm is as follows.
The goal is to compute a genlex listing of all independent sets of~$G$, i.e., independent sets with the same suffix must appear consecutively.
Thus, the algorithm greedily changes the current independent set~$x$ in such a way that a new independent set~$x'$ is reached, i.e., one that has not been visited before, such that~$x'$ agrees with~$x$ in the longest possible suffix, or equivalently, such that $x'$ differs from~$x$ in the shortest possible prefix.
The length~$k$ of the shortest possible prefix of~$x$ with this property is computed (somewhat magically) in step~I3.
Once we know~$k$, all the remaining steps are trivial:
Step~I4 starts by complementing the $k$th bit in~$x$.
If $x_k=1$ before the update, hence if vertex~$k$ is in the current independent set, then it is removed and we have $x_k=0$ after the update.
No further updates are needed, because the new set is clearly independent.
On the other hand, if $x_k=0$ before the update, hence if vertex~$k$ is not in the current independent set, then it is added and we have $x_k=1$ after the update.
Consequently, all neighbors of~$k$ in~$G$ that are in the independent set have to be removed in order to obtain a valid independent set again, which happens in the second part of step~I4.
Recall that $N(k)=N_<(k)\cup N_>(k)$, and since all vertices $j\in N_>(k)$ satisfy $x_j=0$ by the definition of~$k$ in step~I3, only the bits at positions $i\in N_<(k)$ have to be set to~0 (they may already be~0, in which case they do not change).
We see that the next independent set differs from the previous one by a star exchange.

We now explain the magic of the auxiliary array~$b$ that allows us to compute which prefix changes lead to a new independent set, without actually keeping track of the visited ones (this would be prohibitively space- and time-expensive).
Specifically, we have the following invariant: If $b_i=0$, then all independent sets of~$G$ with suffix~$\ol{x_i},x_{i+1},\ldots,x_n$ have been visited before.
Therefore, prefixes of length~$i$ of~$x$ with $b_i=0$ are not eligible to be modified.
On the other hand, prefixes of length~$i$ with $b_i=1$ are eligible, provided that all neighbors~$j\in N_>(i)$ are not in the independent set (otherwise the change would extend beyond the prefix of length~$i$).
This justifies the computation of~$k$ in step~I3 as the shortest possible prefix length that leads to a new independent set.
In particular, if no such length exists ($k=\infty$) the algorithm terminates.
It also explains the update in step~I5, namely, $b_k$ is toggled from $b_k=1$ to $b_k=0$, and all earlier $b_i$ with $i<k$ are set to~1.
The initialization $b\leftarrow 1^n$ in step~I1 indicates that all prefixes are eligible to be modified, as no independent set has been visited yet.

Figure~\ref{fig:protocol} shows two runs of the algorithm on two different input graphs.
We encourage the reader to work through these examples in order to digest and appreciate the working of the algorithm (the array~$s$ in the figure will be introduced later).

\begin{figure}
\includegraphics{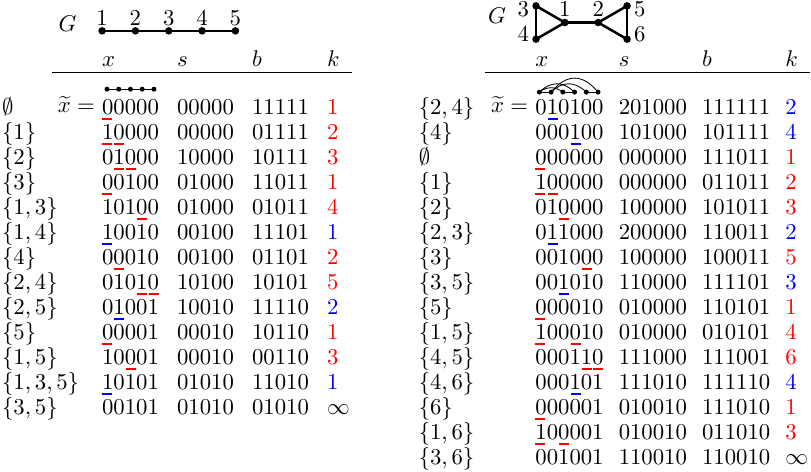}
\caption{Protocol of Algorithms~I and~I\sss{} for two input graphs, one of them being the 5-vertex path from Figure~\ref{fig:ind}, resulting in the Hamilton path shown there.
Both algorithms produce the same results, and the auxiliary array~$s$ is only used within Algorithm~I\sss{}.
All toggled bits of~$x$ are underlined.
Steps marked blue are those where $x_k$ changes from~1 to~0, i.e., vertex~$k$ is removed from the independent set.
Steps marked red are those where $x_k$ changes from~0 to~1, i.e., vertex~$k$ is added to the independent set, and all of its (smaller) neighbors are removed.
}
\label{fig:protocol}
\end{figure}

\begin{theorem}
\label{thm:algo-basic}
For every graph~$G=([n],E)$ and every initial independent set~$\tx\in X(G)$, Algorithm~I computes a genlex Hamilton path on the independent set polytope~$\scP(G)$ starting at~$\tx$, where any two consecutive independent sets differ in a star exchange.
The worst-case delay is~$\cO(n(d_>(G)+1))=\cO(n(\Delta(G)+1))=\cO(n^2)$.
\end{theorem}

We see that choosing a good labeling of the vertices of the input graph is essential for obtaining a good runtime bound for the algorithm.
As emphasized before, Algorithm~I works for any (re)labeling of the vertices of the underlying graph.
For graphs with bounded degeneracy, in particular with bounded maximum degree, the delay is linear in~$n$.

\begin{proof}
We first establish the correctness of the algorithm, and then the bound on the running time.

When describing the algorithm before, we have already argued that it performs only star exchanges, and that every~$x\in\{0,1\}^n$ visited indeed encodes an independent set of~$G$.
It remains to argue that the algorithm visits each independent set of~$G$ exactly once, hence that none of them is missed and none of them is visited repeatedly, and that the produced listing has the genlex property.
This argument uses induction on~$n$, the number of vertices of~$G$.
To settle the induction basis~$n=1$, consider the 1-vertex graph~$G=([1],\emptyset)$, which has two independent sets, namely $\cI(G)=\{\emptyset,\{1\}\}$, encoded by $x = x_1=0$ and~$x = x_1=1$.
One can check that for either of the two possible initializations~$\tx=0$ or $\tx=1$ the algorithm first visits~$\tx$, then its complement, and then terminates, as it should.
This settles the induction basis.

For the induction step, we assume that $n\geq 2$, and let $\tx\in X(G)$ be an arbitrary initial independent set of~$G=([n],E)$.
We first consider the case that $\tx_n=0$.
Observe that as long as the current independent set~$x$ satisfies $x_n=0$, the algorithm operates as if the last bit of~$x$ was not present.
This is because $x_n=0$ does not contribute to the sum $\sum_{j\in N>(i)} x_j$ in the definition of~$k$ (step~I3).
Thus, the algorithm behaves as it would on the input graph~$G-n$, i.e., the graph obtained from~$G$ by removing the last vertex~$n$.
Therefore, by induction, Algorithm~I will visit each independent set of~$G-n$ exactly once, thus producing each independent set of~$G$ that does not contain the vertex~$n$ exactly once.
Let $x$ be the last independent set visited in step~I2 in this first phase of the algorithm, just before the algorithm would terminate in step~I3 if the bit~$b_n$ was not present.
Now that the last bit~$b_n=1$ (recall step~I1) is present, the algorithm does not terminate, but computes $k\leftarrow n$ as the shortest possible prefix to modify.
Consequently, $x_n$ is flipped from~0 to~1, $b_n$ is flipped from~1 to~0, and all $b_1,\ldots,b_{n-1}$ are reset to~$1^{n-1}$.
At the end of this iteration of the main loop~$x$ is the first independent set of~$G$ that contains the vertex~$n$ (to be visited after the goto).
This is the beginning of the second phase of the algorithm.
In the second phase, while $x_n=1$, step~I3 will never choose~$k$ as one of the neighbors $i\in N_<(n)=N(n)$ of~$n$, and the corresponding bits of~$x$ will remain~0 and not be modified (also, the corresponding bits of~$b$ will remain~1 and not be modified).
As a result, we can ignore those bits, and observe that the algorithm behaves as it would on the input graph~$G-N^+(n)$, that is, the graph obtained from~$G$ by removing the closed neighborhood of~$n$.
Therefore, again by induction, Algorithm~I will visit each independent set of~$G-N^+(n)$ exactly once, thus producing each independent set of~$G$ that does contain the vertex~$n$ exactly once.
By~\eqref{eq:IG-rec}, Algorithm~I thus visits each independent set of~$G$ exactly once.
All $x\in X(G)$ with $x_n=0$ are visited before all those with~$x_n=1$, and by induction this property holds recursively within each block obtained by ignoring the same last bit.
It follows that the listing produced by the algorithm is indeed genlex.

To complete the correctness proof, note that the other case $\tx_n=1$ is proved analogously, by swapping the order of the two phases discussed before.

It remains to discuss the running time of the algorithm.
We conveniently assume that the graph~$G=([n],E)$ is represented by adjacency lists.
We write $m\coloneq |E(G)|$ for the number of edges of~$G$.
Once in the beginning, we precompute the sets~$N_<(i)$ and~$N_>(i)$ for all~$i\in[n]$ in time~$\cO(m)$.
In addition, the initialization in step~I1 takes time~$\cO(n)$.
Using that $m\leq n\cdot d_>(G)$ gives the bound $\cO(n(d_>(G)+1))$ for the time until the first vertex is visited (the +1 prevents that the product is~0 in the exceptional case that~$d_>(G)=0$).

Now consider the main loop.
As $|N_<(k)|<n$ and $k\leq n$, steps~I4 and~I5 can clearly be computed in time~$\cO(n)$.
The bottleneck of the computation in the main loop is step~I3.
Indeed, each of the sums~$\sum_{j\in N_>(i)}x_j$ in the minimization to determine~$k$ can be computed in time~$d_>(G)$, and there are at most $n$ such sums to compute, which yields the bound~$\cO(n (d_>(G)+1))$.
\end{proof}

Note that in an actual implementation we do not need to compute the entire sums~$\sum_{j\in N_>(i)}x_j$ in step~I3, but can shortcut once the first non-zero summand~$x_j$ is encountered, which immediately makes the sum nonzero, and thus allows us to discard the corresponding index~$i$ in the minimization.
However, this shortcut does not improve the theoretical guarantees proved before.

\subsection{Improved algorithm}

In the analysis before we noted that the bottleneck in the running time of Algorithm~I is the minimization to determine~$k$ in step~I3.
This directly motivates the introduction of an auxiliary array of integers $s=(s_1,\ldots,s_n)$, where
\begin{equation}
\label{eq:si}
s_i\coloneq \sum\nolimits_{j\in N_>(i)}x_j
\end{equation}
for $i\in[n]$.
With those values at hand, the minimization simply becomes $\min\{i\in[n]\mid b_i=1\wedge s_i=0\}$, which can be computed in linear time if the~$s_i$ are known.
However, introducing the array~$s$ comes at the cost of updating it once any bits of~$x$ are modified.
The required updates are straightforward, see lines~I40 and~I41 in the resulting pseudocode stated as Algorithm~I\sss{} below.

\begin{algo}{Algorithm~I\sss{}}{Independent sets by star exchanges on shortest prefixes}
For a graph~$G=([n],E)$, this algorithm computes a Hamilton path on the independent set polytope~$\scP(G)$, starting from an initial independent set~$\tx\in X(G)$.
\begin{enumerate}[label={\bfseries I\arabic*.}, leftmargin=8mm, noitemsep, topsep=3pt plus 3pt]
\item{} [Initialize] Set $x \gets \tx$, $s_i\gets \sum_{j\in N_>(i)}x_j$ for $i=1,\ldots,n$, and $b\gets 1^n$.
\item{} [Visit] Visit~$x$.
\item{} [Shortest prefix change] Compute $k\gets \min\{i\in[n]\mid b_i=1\wedge s_i=0\}$.
Terminate if $k=\infty$.
\item{} [Update $x$ and $s$] Set $x_k\gets \ol{x_k}$. \\
{\bf I40.} [Flipped $1\rightarrow 0$] If $x_k=0$, then for all $i\in N_<(k)$ set $s_i\gets s_i-1$. \\
{\bf I41.} [Flipped $0\rightarrow 1$] If $x_k=1$, then for all $i\in N_<(k)$ set $s_i\gets s_i+1$, and if $x_i=1$ set $x_i\gets 0$ and for all $j\in N_<(i)$ also set $s_j\gets s_j-1$.
\item{} [Update $b$] Set $b_k\gets 0$, and for $i=1,\ldots,k-1$ set $b_i\gets 1$, then goto~I2.
\end{enumerate}
\end{algo}

Note that the initialization of~$s$ in step~I1 can be simplified to $s\gets 0^n$ if the initial independent set $\tx=0^n$ is used.

\begin{theorem}
\label{thm:algo-faster}
For every graph~$G=([n],E)$ and every initial independent set~$\tx\in X(G)$, Algorithm~I\sss{} computes the same Hamilton path as Algorithm~I.
After the initialization in time~$\cO(n(d_<(G)+1))$, the worst-case delay is~$\cO(n+d_<(G)^2)=\cO(n+\Delta(G)^2)=\cO(n^2)$.
The amortized delay is~$\cO(n)$.
\end{theorem}

We see again that choosing a good labeling of the vertices of the input graph is important.
However, while Algorithm~I benefits from small forward degrees in the vertex ordering of~$G$, captured by the forward degeneracy~$d_>(G)$, Algorithm~I\sss{} benefits from small backward degrees, captured by the backward degeneracy~$d_<(G)$.
Importantly, the amortized delay of Algorithm~I\sss{}, which matters most, is linear in~$n$, regardless of the degrees in~$G$.

An experimental comparison of both algorithms on a variety of input graphs, exploring in particular the issue of vertex orderings, is presented in Section~\ref{sec:exp}.

\begin{proof}
The first part of the statement follows from the correctness of Algorithm~I and the definition~\eqref{eq:si} of the entries of the array~$s$.

We now argue about the claimed delay bounds.
We first consider the worst-case delay.
Precomputing the sets~$N_<(i)$ and~$N_>(i)$ for all~$i\in[n]$ and the initialization in step~I1 can be done in time~$\cO(m)$.
Using that $m\leq n\cdot d_<(G)$ gives the bound $\cO(n(d_<(G)+1))$ for the time until the first vertex is visited.
Now consider the main loop.
Clearly, steps~I3, I40 and I5 can be computed in time~$\cO(n)$.
Step~I41 contains a double loop that takes time~$\cO(d_<(G)^2)$.
Adding these contributions proves the delay bound.

We now consider the amortized delay.
The time $\cO(n(d_<(G)+1))$ until the first vertex is visited can be amortized to~$\cO(d_<(G))=\cO(n)$ by the observation that~$G$ has at least $n$ distinct independent sets (containing exactly one vertex only).
It remains to argue about the amortization inside the double loop of step~I41.
Whenever a 0-bit at position~$k$ in~$x$ is toggled to a 1-bit, the outer loop of step~I41 (`for all $i\in N_<(k)$ set $s_i\gets s_i+1$') is executed.
We can thus, at no asymptotic cost, prepay costs for the 1-bit~$x_k$ in the amount of~$|N_<(k)|$: Think about placing 1~coin on each edge going from~$k$ to a smaller neighbor.
Now observe that that the inner loop in step~I41 (`for all $j\in N_<(i)$ also set $s_j\gets s_j-1$') is executed only for positions~$i\in[n]$ for which~$x_i=1$.
These are exactly the costs that have been prepaid: Think about taking back the previously placed coins from the edges going from~$i$ to a smaller neighbor while iterating through the loop.
The 1-bits in the initial string~$\tx$ are prepaid for in the initialization of~$s$ in step~I1.
The amortized cost of one iteration of the main loop is therefore~$\cO(n+d_<(G))=\cO(n)$.
\end{proof}

\section{Applications}
\label{sec:appl}

\subsection{Independent sets of special families of graphs}
\label{sec:families}

We first apply our algorithms to compute independent sets of special families of graphs, which recovers several Gray codes known from the literature.

We need the following definitions.
For any sequence of bitstrings~$L=(x_1,\ldots,x_N)$, we write $\rev(L)\coloneq (x_N,x_{N-1},\ldots,x_1)$ for the reverse sequence.
Furthermore, for any two bitstrings~$u,v$ we define $uLv\coloneq (ux_1v,ux_2v,\ldots,ux_Nv)$, where the prefix~$u$ is prepended and suffix~$v$ is appended to every bitstring in~$L$.
For a bitstring~$x$ and integer~$n\geq 1$ we write $[x]_n$ for the length~$n$ substring of the left-concatenation of infinitely many copies of~$x$.
In other words, $[x]_n$ is obtained by prepending~$x$ with more copies of~$x$ until all $n$ positions are filled, truncating the extra bits at the beginning.
For example, we have $[100]_5=00100$ and $[100]_{10}=0100100100$.

We write $E_n$ for the graph with vertex set~$[n]$ and no edges.
Clearly, we have $\cI(E_n)=2^{[n]}$, i.e., every possible subset of~$[n]$ is independent.
The listing of independent sets obtained when running Algorithms~I and~I\sss{} with input graph~$E_n$ and initial independent set~$\tx=0^n$ is the well-known \defi{binary reflected Gray code}, defined inductively by
\[ B_1=0,1 \;\;\text{and}\;\; B_n\coloneq B_{n-1}0,\rev(B_{n-1})1 \text{ for } n\geq 2; \]
see Figure~\ref{fig:listings}~(a).
I.e., we obtain a listing of all bitstrings of length~$n$, such that any two consecutive strings differ in a single bit.
In the following, we use the term \defi{1-Gray code} for such a listing of bitstrings.
In terms of independent sets, this further restricts the allowed exchanges to removing or adding only a single vertex in every step (i.e., it is never necessary to remove the neighbors of an added vertex from the independent set).

\begin{figure}[b!]
\centerline{
\setlength{\tabcolsep}{3pt}
\renewcommand{\arraystretch}{1.5}
\begin{tabular}{ccccccccccc}
(a) $B_2$ & $B_3$ & $B_4$ & $B_5$ & (b) $F_4$ & $F_5$ & $F_6$ & $F_7$ & (c) $L_4$ & $L_5$ & $L_7$ \\[-5mm]
\phantom{(a)}\raisebox{-\height}{\includegraphics{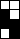}} &
\raisebox{-\height}{\includegraphics{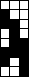}} &
\raisebox{-\height}{\includegraphics{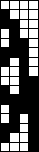}} &
\raisebox{-\height}{\includegraphics{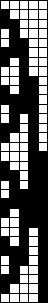}} &
\phantom{(b)}\raisebox{-\height}{\includegraphics{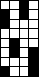}} &
\raisebox{-\height}{\includegraphics{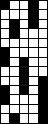}} &
\raisebox{-\height}{\includegraphics{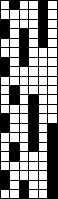}} &
\raisebox{-\height}{\includegraphics{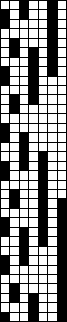}} &
\phantom{(c)}\raisebox{-\height}{\includegraphics{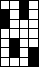}} &
\raisebox{-\height}{\includegraphics{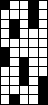}} &
\raisebox{-\height}{\includegraphics{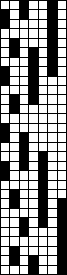}}
\end{tabular}
}
\caption{1-Gray codes produced by Algorithms~I and~I\sss{} for various special input graphs (0-bits=white squares, 1-bits=black squares), namely (a)~empty graphs $E_2,E_3,E_4,E_5$ (binary reflected Gray code), (b)~path graphs $P_4,P_5,P_6,P_7$ (Fibonacci words), and (c)~cycle graphs $C_4,C_5,C_7$ (Lucas words).}
\label{fig:listings}
\end{figure}

We write $P_n$ and $C_n$ for the path graph and cycle graph with vertex set~$[n]$, with vertices labeled~$1,\ldots,n$ along the path or cycle, respectively.
Clearly, independent sets of~$P_n$ and~$C_n$ are subsets of~$[n]$ that contain no two consecutive vertices or no two cyclically consecutive vertices, respectively;
in the latter case vertex~1 and vertex~$n$ are also considered consecutive.
The corresponding characteristic vectors have been studied in the pattern-avoidance community.
Specifically, \defi{Fibonacci words} are bitstrings that avoid two consecutive 1s, and \defi{Lucas words} are bitstrings that avoid two cyclically consecutive 1s.
Parametrizing by the length~$n$ of the words, they are counted by the Fibonacci numbers and Lucas numbers, respectively (OEIS sequences~A000045 and A000032).

We inductively define a listing~$F_n$ of all Fibonacci words of length~$n$ by
\begin{equation}
\label{eq:Fn}
F_1\coloneq 0,1 \;\;\text{and}\;\; F_n\coloneq \rev(F_{n-1})0,\rev(F_{n-2})01 \text{ for } n\geq 2;
\end{equation}
see Figure~\ref{fig:listings}~(b) (cf.~\cite{MR1934834} and~\cite{mikawa_semba_2005}).
Similarly, for all $n$ that are not multiples of~3 we define a listing~$L_n$ of all Lucas words of length~$n$ by
\begin{equation}
\label{eq:Ln}
L_1\coloneq 0,1, L_2\coloneq 10,00,01 \;\;\text{and}\;\; L_n\coloneq \rev(F_{n-1})0,0\rev(F_{n-3})01 \text{ for } n\geq 4 \text{ with }3\nmid n;
\end{equation}
see Figure~\ref{fig:listings}~(c) (cf.~\cite{MR2187406}).
Maybe slightly unexpectedly, the definition of~$L_n$ recurses to~$F_{n-1}$ and~$F_{n-3}$, \emph{not} to $L_{n-1}$ and~$L_{n-3}$.
Using these definitions, one can check easily by induction that both listings~$F_n$ and~$L_n$ are genlex 1-Gray codes, where $F_n$ starts with $[010]_n$ and ends with $[001]_n$, and $L_n$ starts with $[010]_n$ and ends with~$0[001]_{n-1}$.
We note that Figure~\ref{fig:ind}~(c) also shows a listing of Fibonacci words of length~5, but the initial bitstring is different and in some steps, more than one bit changes, so the listing is not a 1-Gray code.

\begin{theorem}
\label{thm:fib-luc}
Running Algorithms~I or~I\sss{} with input graph~$G=P_n$ for $n\geq 1$ or~$G=C_n$ for $n\geq 1$ with $3\nmid n$ and initial independent set~$\tx\coloneq [010]_n$ yields the genlex 1-Gray codes~$F_n$ and~$L_n$ of Fibonacci and Lucas words defined in~\eqref{eq:Fn} and~\eqref{eq:Ln}, respectively.
\end{theorem}

\begin{proof}
The following auxiliary claim can be proved by a straightforward induction on~$n$:
For any graph~$G=([n],E)$ and any $\tx\in X(G)$, there is a unique genlex listing of~$X(G)$ by star exchanges starting at~$\tx$.
Flipping a single bit is clearly a star exchange, so the 1-Gray codes $F_n$ and~$L_n$ are genlex listings by star exchanges.
From Theorem~\ref{thm:algo-basic} and~\ref{thm:algo-faster} and the auxiliary claim it follows that Algorithms~I and~I\sss{} compute exactly those listings.
\end{proof}

\subsection{Vertex Covers}
\label{sec:cover}

A \defi{vertex cover} of a graph~$G=([n],E)$ is a subset of vertices~$U\seq [n]$, such that every edge of~$G$ is incident with at least one vertex in~$U$.
It is easy to see that $U\seq [n]$ is a vertex cover if and only if its complement~$[n]\setminus U$ is an independent set.
Consequently, the \defi{vertex cover polytope}~$\overline{\scP}(G)\coloneq \conv\{x^U\mid U\text{ vertex cover of }G\}$ is obtained by applying the reflection $x\mapsto 1-x$ to the independent set polytope~$\scP(G)$, and thus both polytopes have the same combinatorial structure.
Both Algorithm~I and Algorithm~I\sss{} can be used to compute a Hamilton path on the vertex cover polytope~$\overline{\scP}(G)$, by computing the complement~$\overline{x}$ just before every visit step.
Both algorithms compute the same genlex Hamilton path on the vertex cover polytope~$\overline{\scP}(G)$ within the time bounds guaranteed by Theorems~\ref{thm:algo-basic} and~\ref{thm:algo-faster}, where any two consecutive vertex covers differ in a complemented star exchange, i.e., either a single vertex is added to the cover, or a vertex is removed and its neighbors are added.

\subsection{Matchings}

Matchings of a graph are independent sets in its line graph, so we can apply our algorithms to compute a Hamilton path on the corresponding polytope.
We proceed to translate our algorithm to this setting.

A \defi{matching} in a graph~$G$ is a subset of edges of~$G$, no two of which have a vertex in common.
Given a graph~$G$, its \defi{line graph} $L(G)$ has a vertex for every edge of~$G$, and an edge between any two vertices corresponding to edges of~$G$ that are incident to each other.
Note that the line graph~$L(G)$ is \defi{claw-free}, that is, it has no induced star with three rays.
This expresses the fact that an edge~$\{u,v\}$ in~$G$ can be incident with at most two other edges that have no end vertices in common with each other (one incident with~$u$ and one incident with~$v$).
We observe that matchings in~$G$ are in one-to-one correspondence with independent sets in the line graph~$L(G)$.

We write $\cM(G)$ for the set of all matchings of~$G$ and $X(G)$ for the corresponding set of characteristic vectors, i.e., $X(G)\coloneq \{x^M\mid M\in\cM(G)\}$.
The \defi{matching polytope} is defined $\scM(G)\coloneq \conv(X(G))=\conv(\{x^M\mid M\text{ matching in }G\})$~\cite{MR371732}.
By Lemma~\ref{lem:edges-polytope}, for any two distinct matchings~$M,N\in\cM(G)$, the pair of vertices~$\{x^M,x^N\}$ is an edge of the polytope~$\scM(G)$ if and only if the symmetric difference $M\triangle N$ is a single path or cycle.

\begin{wrapfigure}{r}{0.35\textwidth}
\centering
\includegraphics{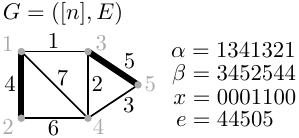}
\caption{Data structures used in Algorithm~M. Matching edges are bold.}
\label{fig:match}
\end{wrapfigure}
We consider a graph~$G=([n],E)$ with $m$ edges.
We formalize the labeling of the edges of~$G$ by integers~$1,\ldots,m$ by two mappings $\alpha,\beta:[m]\rightarrow [n]$, where $\alpha(i)$ and $\beta(i)$ are the two end vertices of the edge with label~$i$.
We thus refer to the pair of mappings $\alpha,\beta$ as an \defi{edge labeling}.
We translate Algorithm~I to the setting of matchings as follows:
The array $x=(x_1,\ldots,x_m)$ is the characteristic vector of the current matching, and the array~$b=(b_1,\ldots,b_m)$ has the same meaning as before.
In addition, we introduce an array $e=(e_1,\ldots,e_n)$ where $e_i\in[m]\cup\{0\}$, with the following meaning: $e_j=i\in[m]$ means that vertex~$j$ is incident with edge~$i$ of the matching, whereas $e_j=0$ means that vertex~$j\in[n]$ is not incident with any edge of the matching; see Figure~\ref{fig:match}.
With these data structures, we obtain Algorithm~M stated below.
The condition~$(x_i=1\vee \max\{e_{\alpha(i)},e_{\beta(i)}\}<i)$ in line~M2 is satisfied if $x_i=1$, i.e., edge~$i$ is currently in the matching, or if $x_i=0$, i.e., edge~$i$ is not currently in the matching, and both matching edges incident to either~$\alpha(i)$ or~$\beta(i)$ have smaller labels (this includes the case that there is no such edge, because then the special label~0 is used).
The star exchanges of Algorithm~I thus translate to either removing an edge from the matching or adding an edge and removing at most two other edges.
Formally, for two matchings~$M,N\in\cM(G)$, the transition $M\rightarrow N$ is called a \defi{$\{1,2,3\}$-exchange}, if either $N=M\setminus\{e\}$ for some $e\in M$ or $N=(M\cup\{e\})\setminus F$ for some~$e\notin M$ and subset $F\seq M$ such that $e$ shares an endpoint with every edge $f\in F$ (thus $|F|\leq 2$).
In both cases, $M\triangle N$ is a path of length at most~3.

\begin{algo}{Algorithm~M}{Matchings by shortest prefix exchanges}
For an $m$-edge graph~$G=([n],E)$ with an edge labeling $\alpha,\beta\in [n]^{[m]}$, this algorithm computes a Hamilton path on the matching polytope~$\scM(G)$, starting from an initial matching $x\in X(G)$.
\begin{enumerate}[label={\bfseries M\arabic*.}, leftmargin=9mm, noitemsep, topsep=3pt plus 3pt]
\item{} [Initialize] Set $x \gets \tx$, $b\gets 1^m$ and $e\gets 0^n$.
For $i=1,\ldots,m$, if $x_i=1$ set $e_{\alpha(i)}\gets i$ and $e_{\beta(i)}\gets i$.
\item{} [Visit] Visit~$x$.
\item{} [Shortest prefix change] Compute $k\gets \min\{i\!\in\![m]\mid b_i\!=\!1\wedge (x_i\!=\!1\vee \max\{e_{\alpha(i)},e_{\beta(i)}\}\!<\!i)\}$.
Terminate if $k=\infty$.
\item{} [Update $x$ and~$e$] Set $x_k\gets \ol{x_k}$. \\
{\bf M40.} [Flipped $1\rightarrow 0$] If $x_k=0$, then set $e_{\alpha(k)}\gets 0$ and $e_{\beta(k)}\gets 0$. \\
{\bf M41.} [Flipped $0\rightarrow 1$] If $x_k=1$, then for each $j\in \{\alpha(k),\beta(k)\}$, if $e_j>0$ set $x_{e_j}\gets 0$, $e_{\alpha(e_j)}\gets0$ and $e_{\beta(e_j)}\gets0$. Then set $e_j\gets k$.
\item{} [Update $b$] Set $b_k\gets 0$, and for all $i=1,\ldots,k-1$ set $b_i\gets 1$, then goto~M2.
\end{enumerate}
\end{algo}

\begin{theorem}
\label{thm:matching}
For every $m$-edge graph~$G=([n],E)$ and every initial matching~$\tx\in X(G)$, Algorithm~M computes a genlex Hamilton path on the matching polytope~$\scM(G)$ starting at~$\tx$, where any two consecutive matchings differ in a $\{1,2,3\}$-exchange.
The worst-case delay is~$\cO(m+n)$.
\end{theorem}

\begin{proof}
The correctness of the algorithm follows from Theorem~\ref{thm:algo-basic}, and the running time can be seen directly.
\end{proof}

\subsection{Antichains and ideals}
\label{sec:anti-ideal}

Antichains of a poset are independent sets in its comparability graph, so we can apply our algorithms to compute a Hamilton path on the corresponding polytope.
We now set up the required notation to translate all notions and results to the setting of posets.

Given a poset~$P=([n],<)$, we write $D(i)\coloneq \{j\in[n]\mid j\leq i\}$ and $U(i)\coloneq \{j\in[n]\mid i\leq j\}$ for the \defi{downset of~$i$} and \defi{upset of~$i$}, respectively.
We also define the \defi{proper downset} and \defi{proper upset} as $D_<(i)\coloneq D(i)\setminus\{i\}=\{j\in[n]\mid j<i\}$ and $U_>(i)\coloneq U(i)\setminus\{i\}=\{j\in[n]\mid i<j\}$, respectively.
The \defi{comparability graph} of~$P$ is the graph with vertex set~$[n]$ and an edge~$\{i,j\}$ if either $i<j$ or $j<i$ in~$P$, i.e., it has edges between all comparable pairs of elements in the poset.
We say that $P$ is \defi{connected} if its comparability graph is connected.
For a subset $A\seq [n]$ we write $P[A]$ for the subposet induced by~$A$, consisting of all elements in~$A$ with the comparabilities given by~$P$.

A \defi{chain} or \defi{total order} is a poset in which every two elements are comparable.
An \defi{antichain} is a poset in which no two elements are comparable.
For a poset~$P$, we write $\cA(P)$ for the set of all antichains in~$P$ and~$X(P)$ for the corresponding set of characteristic vectors, i.e., $X(P)\coloneq \{x^A\mid A\in\cA(P)\}$.
Note that antichains in~$P$ are in one-to-one correspondence with independent sets in the comparability graph of~$P$.
The \defi{chain polytope} of~$P$ is defined as $\scC(P)\coloneq \conv(X(P))=\conv(\{x^A\mid A\text{ antichain in }P\})$~\cite{MR824105}\footnote{The motivation for the name chain polytope (instead of antichain polytope) is that its hyperplane representation is based on chains in the poset. However, in our context the vertex representation is more convenient.}.
By Lemma~\ref{lem:edges-polytope}, for any two distinct antichains~$A,B\in\cA(P)$, the pair of vertices~$\{x^A,x^B\}$ is an edge of the polytope~$\scC(P)$ if and only if the subposet $P[A\triangle B]$ is connected.
In particular, a star exchange between two independent sets in the comparability graph corresponds to either removing an element from the antichain, or adding an element and removing its proper downset or proper upset.

\begin{wrapfigure}{r}{0.45\textwidth}
\centering
\includegraphics[page=1]{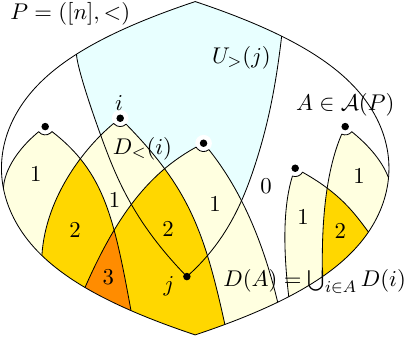}
\caption{An antichain~$A$ and its corresponding ideal~$D(A)$.
The numbers in the colored regions indicate the number of antichain elements in the upset of an element, which are stored in the variables~$s_i$ (cf.~\eqref{eq:si-poset}).}
\label{fig:antichain}
\end{wrapfigure}
A \defi{linear extension} of a poset~$P=([n],<_P)$ is a total order~$L=([n],<_L)$ that extends~$P$, hence such that if $i<_P j$ then~$i<_L j$.
We can think of a linear extension as a permutation on~$[n]$ such that if $i<_P j$ then the value~$i$ appears before the value~$j$ in the permutation.
We denote the identity permutation by $\id_n\coloneq (1,2,\ldots,n)$.

As mentioned in Section~\ref{sec:results}, our algorithms work for any labeling of the ground set.
In the context of posets, it is convenient to choose a labeling that is a linear extension.
Specifically, we refer to~$P=([n],<)$ as an \defi{upward} poset if it has the identity permutation~$\id_n=(1,2,\ldots,n)$ as a linear extension.
Clearly, any poset can be (re)labeled in this way, by numbering the elements from the bottom up.
With this labeling, the set of left neighbors~$N_<(i)$ and the set of right neighbors~$N_>(i)$ in the comparability graph simply become the proper downset~$D_<(i)$ and proper upset~$U_>(i)$, respectively.
Consequently, the entries of the array~$s$ defined in~\eqref{eq:si} have the interpretation
\begin{equation}
\label{eq:si-poset}
s_i=\sum\nolimits_{j\in U_>(i)} x_j=|U_>(i)\cap A|,
\end{equation}
where $A\in\cA(P)$ denotes the antichain whose characteristic vector is~$x$, i.e., $x=x^A$.
In words, $s_i$ counts the number of elements of the antichain~$A$ in the proper upset of~$i$, or equivalently, the number of elements of~$A$ that contain~$i$ in their proper downset; see Figure~\ref{fig:antichain}.
Furthermore, for an upward poset, a star exchange between two independent sets in the comparability graph corresponds to either removing an element from the antichain, or adding an element and removing its proper downset (the operation of removing the proper upset does not occur).
Thus, for two antichains~$A,B\in\cA(P)$, the transition $A\rightarrow B$ is called a \defi{star exchange} if either $B=A\setminus \{k\}$ for some~$k\in A$ or $B=(A\cup\{k\})\setminus D_<(k)$ for some $k\notin A$; see Figure~\ref{fig:star-poset}~(a).

\begin{figure}[t!]
\includegraphics[page=2]{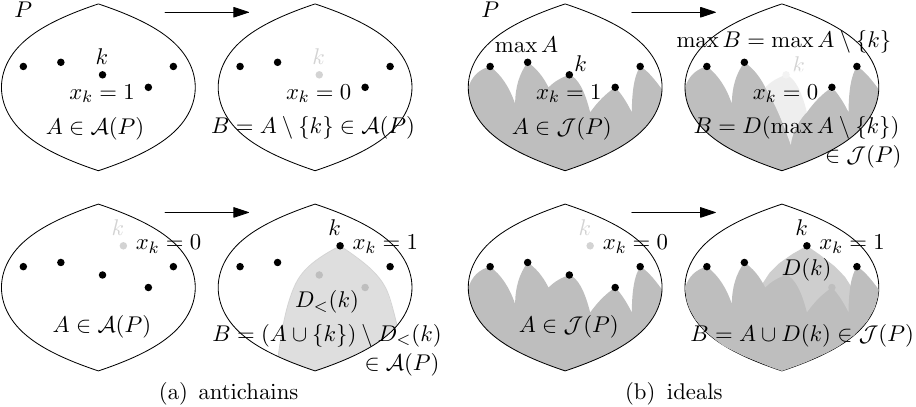}
\caption{Star exchanges for (a) antichains and (b) ideals of posets.}
\label{fig:star-poset}
\end{figure}

\begin{algo}{Algorithm~A\sssm{}}{Antichains by star exchanges on shortest prefixes}
For an upward poset~$P=([n],<)$, this algorithm computes a Hamilton path on the chain polytope~$\scC(P)$, starting from an initial antichain~$\tx\in X(P)$.
\begin{enumerate}[label={\bfseries A\arabic*.}, leftmargin=8mm, noitemsep, topsep=3pt plus 3pt]
\item{} [Initialize] Set $x \gets \tx$, $s_i\gets \sum_{j\in U_>(i)}x_j$ for $i=1,\ldots,n$, and $b\gets 1^n$.
\item{} [Visit] Visit~$x$.
\item{} [Shortest prefix change] Compute $k\gets \min\{i\in[n]\mid b_i=1\wedge s_i=0\}$.
Terminate if $k=\infty$.
\item{} [Update $x$ and $s$] Set $x_k\gets \ol{x_k}$. \\
{\bf I40.} [Flipped $1\rightarrow 0$] If $x_k=0$, then for all $i\in D_<(k)$ set $s_i\gets s_i-1$. \\
{\bf I41.} [Flipped $0\rightarrow 1$] If $x_k=1$, then for all $i\in D_<(k)$ set $s_i\gets s_i+1$, and if $x_i=1$ set $x_i\gets 0$ and for all $j\in D_<(i)$ also set $s_j\gets s_j-1$.
\item{} [Update $b$] Set $b_k\gets 0$, and for $i=1,\ldots,k-1$ set $b_i\gets 1$, then goto~A2.
\end{enumerate}
\end{algo}

The algorithm obtained by translating Algorithm~I\sss{} to the poset setting is stated as Algorithm~A\sssm{}.
Of course, the more basic but slower variant analogous to Algorithm~I, without the array~$s$, also exists, but is omitted.

The initialization of~$s$ in step~A1 can be simplified to $s\gets 0^n$ if the initial antichain $\tx=0^n$ is used.

From Theorems~\ref{thm:algo-basic} and~\ref{thm:algo-faster} we thus derive the following result.

\begin{theorem}
\label{thm:antichain}
For every upward poset~$P=([n],<)$ and every initial antichain~$\tx\in X(P)$, Algorithm~A\sssm{} computes a genlex Hamilton path on the chain polytope~$\scC(P)$ starting at~$\tx$, where any two consecutive antichains differ in a star exchange.
The amortized delay is~$\cO(n)$.
\end{theorem}

For the amortized delay stated in the theorem we assume that the poset~$P=([n],<)$ is represented as a directed graph, with separate adjacency lists for the in- and out-neighbors of each element~$i\in[n]$, corresponding to all elements of~$D_<(i)$ and~$U_>(i)$, respectively.

Given a poset~$P=([n],<)$, an element~$i\in[n]$ is \defi{maximal} if there is no $j\in[n]$ with~$i<j$.
An \defi{ideal} is a downward closed set~$A\seq [n]$, hence such that if $i\in A$ and $j<i$ then we have $j\in A$.
We write $\cJ(P)$ for the set of all ideals of~$P$ and $Y(P)$ for the corresponding set of characteristic vectors, i.e., $Y(P)\coloneq \{x^A\mid A\in\cJ(P)\}$.
There is a one-to-one correspondence between antichains and ideals in~$P$, namely, the maximal elements of an ideal form an antichain, and conversely, the union of downsets of all elements in an antichain forms an ideal.
Given an ideal~$A\in\cJ(P)$, we write $\max A$ for the antichain determined by the maximal elements of~$A$.
Conversely, given an antichain~$A\in\cA(P)$, we write $D(A)\coloneq \bigcup_{i\in A}D(i)$ for the ideal generated by the antichain~$A$.

The \defi{order polytope} of~$P$ is defined as $\scO(P)\coloneq \conv(Y(P))=\{x^A\mid A\text{ ideal in }P\}$ \cite{MR824105}\footnote{Stanley considers upper ideals (upward closed), whereas we consider lower ideals (downward closed).
However, the two notions are equivalent by turning the poset upside down.}.
For two distinct ideals~$A,B\in\cJ(P)$, the pair of vertices~$\{x^A,x^B\}$ is an edge of the polytope~$\scO(P)$ if and only if $P[A\triangle B]$ is a connected subposet~\cite{MR824105}.
Note that this condition implies~$A\subsetneq B$ or~$B\subsetneq A$.
Due to the one-to-one correspondence between antichains and ideals, both polytopes~$\scC(P)$ and~$\scO(P)$ have the same number of vertices.
In general, edges and higher-dimensional faces of both polytopes are different (see \cite[Ex.~1.2]{MR3612430}).
Nonetheless, edges on the order polytope coming from star exchanges do translate to edges on the chain polytope, as follows:
We say that two ideals~$A,B\in\cJ(P)$ differ in a \defi{star exchange} if the corresponding antichains~$\max A$ and~$\max B$ differ in a star exchange.
In other words, a transition $A\rightarrow B$ is a star exchange if either
\[ B=D(\max A\setminus \{k\})=A\setminus\big\{i\in D(k)\mid U(i)\cap\max A=\{k\}\big\} \]
for some $k\in \max A$, or $B=A\cup D(k)$; see Figure~\ref{fig:star-poset}~(b).

Using the aforementioned correspondence between antichains and ideals, we can thus modify Algorithm~A\sssm{} and compute the ideals corresponding to the antichains along the way.
The resulting Hamilton paths on the chain polytope are in one-to-one correspondence to the Hamilton paths on the order polytope computed by Algorithm~A\sssm{}.

\begin{theorem}
\label{thm:ideal}
For every upward poset~$P=([n],<)$ and every initial ideal~$\ty\in Y(P)$, the modified Algorithm~\mbox{A\sssm{}} computes a genlex Hamilton path on the order polytope~$\scO(P)$ starting at~$\ty$, where any two consecutive ideals differ in a star exchange.
The amortized delay is~$\cO(n)$.
\end{theorem}

\begin{proof}
We add an extra variable~$y\in\{0,1\}^n$ that stores the characteristic vector of the current ideal.
All computations are done with the corresponding antichain~$x\in X(P)$, and the variable~$y$ is only used in the initialization step and the visit step.
Specifically, we replace line~A1 by:

\vspace{1.5mm}
\hspace{-5mm}
\myboxm{{\bf A1.} [Initialize] For $i=1,\ldots,n$, set $x_i\gets 1$ if $\ty_i=1$ and $\sum_{j\in U_>(i)}\ty_j=0$, and $x_i\gets 0$ otherwise. Also set $s_i\gets \sum_{j\in U_>(i)}x_j$ for $i=1,\ldots,n$, and $b\gets 1^n$.}
\vspace{-1mm}

The initialization of~$x$ and~$s$ can be simplified to $x\gets 0^n$ and $s\gets 0^n$ if the initial ideal $\ty=0^n$ is used.

We observe that the ideal~$D(A)$ of an antichain~$A\in\cA(P)$ satisfies
\[ D(A)=\bigcup\nolimits_{i\in A}D(i)=A\cup\bigcup\nolimits_{i\in A}D_<(i)=A\cup\{i\in[n]\mid |U_>(i)\cap A|>0\}. \]
From~\eqref{eq:si-poset} we see that the quantity $|U_>(i)\cap A|$ for the current antichain~$A$, i.e., for $x\in X(P)$ with $x=x^A$, is stored in the variable~$s_i$; recall Figure~\ref{fig:antichain}.
Thus, we replace line~A2 by:

\vspace{1.5mm}
\hspace{-5mm}
\myboxm{{\bf A2.} [Visit] For $i=1,\ldots,n$, set $y_i\gets 1$ if $x_i=1$ or $s_i>0$, and $y_i\gets 0$ otherwise. Visit $y$.}
\vspace{-2mm}

The listing of ideals produced by this modified algorithm is genlex, because $n\in A$ implies that $n\in D(A)$, and $n\notin A$ implies that $n\notin D(A)$.

These additional computations in the visit step take time~$\cO(n)$, so they do not affect the~$\cO(n)$ amortized delay bound.
\end{proof}

\section{Experimental comparison of Algorithms~I and~I\sss{}}
\label{sec:exp}

In this section, we aim to understand how well the bounds on the (amortized) delay established by Theorems~\ref{thm:algo-basic} and~\ref{thm:algo-faster} reflect the `running time in practice' when using Algorithms~I and~I\sss{} to compute independent sets for a variety of input graphs.
In particular, we aim to measure the effect of the vertex ordering of an input graph on the running time.
For this we implemented both algorithms in C++, incorporating the shortcut discussed after the proof of Theorem~\ref{thm:algo-basic}.

The results of the experiments are depicted in the four plots shown in Figure~\ref{fig:exp}.
The quantity~$D/n$ displayed on the vertical axes is defined as follows:
For each of the Algorithms~I (top row) or Algorithm~I\sss{} (bottom row) and a given input graph~$G=([n],E)$, we count the total number~$C$ of read and write accesses to the arrays~$x$ and~$b$, or $x$, $s$ and~$b$, respectively.
The visit step does not contribute to this count.
Thus, the number~$C$ is a proxy for the total work performed by the algorithm for computing the set~$\cI(G)$ of all independent sets of~$G$, which is independent of the choice of implementation language and computing power.
Then the average delay is simply
\begin{equation}
\label{eq:delay}
D\coloneq \frac{C}{|\cI(G)|},
\end{equation}
which is the work performed per independent set.
Thus, the quantity $D/n$ displayed in the figures is the average delay per bit of the vector~$x$ that encodes each independent set.

The horizontal axis in all plots shows the average degree~$\ol{d}(G)\coloneq 2m/n$, where $m=|E|$ is the number of edges of~$G$.
This quantity was chosen because it is independent of the vertex ordering, and yet roughly corresponds to the forward degeneracy~$d_<(G)$ and backward degeneracy~$d_>(G)$, because of the lower bounds $d_<(G),d_>(G)\geq m/n=\ol{d}(G)/2$, which are reasonably tight for graphs that are nearly regular.

The two plots in the left column are for sparse graphs with average degree in the interval~$[0,5]$, and the two plots in the right column are for dense graphs with average degree in the interval~$[5,115]$.
The ranges of values on the vertical axes are the same in the top and bottom row, to allow direct comparison between the algorithms, but they are drastically different in the left and right column, to provide a more fine-grained picture for sparse graphs.

The following graphs were chosen as input graphs for both algorithms.
\begin{itemize}[leftmargin=2ex,itemindent=0ex]
\item As in Section~\ref{sec:families} before, $E_n$ denotes the graph with $n$ vertices and no edges, and $P_n$ and~$C_n$ denote the path and cycle with $n$ vertices.
\item $Q_n$ denotes the $n$-dimensional hypercube with vertex set~$\{0,1\}^n$ and an edge between any two bitstrings that differ in exactly one bit.
\item $K_{a,b,c}$ denotes the complete tripartite graph with vertex sets of sizes~$a,b,c$.
\item $\Kg_{n,k}$ denotes the Kneser graph whose vertices are all $k$-element subsets of~$[n]$, with edges between disjoint sets.
\item $\Ci_{60,13}$ denotes the 60-vertex circulant graph with the 13-element set of displacement values~$\{1,3,5,6,7,8,9,13,15,17,18,19,20\}$.
\item $G_{n,p}$ denotes the binomial random graph with $n$ vertices and edge probability~$p$.
\end{itemize}

For each input graph~$G$, we consider the following five different orderings of its vertex set.
\begin{itemize}[leftmargin=2ex,itemindent=0ex]
\item \texttt{min}: vertex ordering obtained by repeatedly removing a vertex of minimum degree from~$G$; this is the ordering that minimizes the forward degeneracy~$d_>(G)$.
\item \texttt{revmin}: reverse of~\texttt{min}; this is the ordering that minimizes the backward degeneracy~$d_<(G)$.
\item \texttt{max}: vertex ordering obtained by repeatedly removing a vertex of maximum degree from~$G$.
\item \texttt{revmax}: reverse of~\texttt{max}.
\item \texttt{rand}: average of 10 uniformly random vertex orderings.
\end{itemize}

\subsection{Evaluation}

For the smallest and sparsest graphs, Algorithm~I has a small advantage, likely because it does not have to maintain the extra array~$s$, which is useless in this case, whereas for all other graphs, Algorithm~I\sss{} is faster, and the advantage of the extra array~$s$ pays off.
Most importantly, as the average degree grows, the average delay per bit remains constant for Algorithm~I\sss{}, as predicted by Theorem~\ref{thm:algo-faster}, but grows linearly for Algorithm~I, as predicted by Theorem~\ref{thm:algo-basic}.

\addtolength{\tabcolsep}{-1ex}
\begin{figure}[t!]
\centerline{
\begin{tabular}{ccc}
\rotatebox[origin=l]{90}{\hspace{25mm}Algorithm~I}
&
\begin{tikzpicture}
\begin{axis}[
    xlabel={$2m/n$=average degree},
    ylabel={$\frac{D}{n}$},
    ylabel style={rotate=-90,at={(0.05,0.9)}},
    xmin=0, xmax=5,
    ymin=0, ymax=3.5,
    xtick={0,0.5,1,1.5,2,2.5,3,3.5,4,4.5,5},
    ytick={0,0.5,1,1.5,2,2.5,3,3.5},
    ymajorgrids=true,
    xmajorgrids=true,
    grid style=dashed,
    width=8cm,
    height=6cm,
]
\addplot[color=red,mark=*,mark size=1.0pt]  coordinates {(0,1.39)} node[right] {$E_5$};
\addplot[color=red,mark=*,mark size=1.0pt]  coordinates {(0,0.7)} node[right] {$E_{10}$};
\addplot[color=red,mark=*,mark size=1.0pt]  coordinates {(0,0.28)} node[right] {$E_{25}$};
\addplot[color=blue,mark=*,mark size=1.0pt] coordinates {(1.6,1.95)}; 
\addplot[color=blue,mark=x,mark size=1.2pt] coordinates {(1.6,1.95)}; 
\addplot[color=blue,mark=+,mark size=1.2pt] coordinates {(1.6,2.26)} node[above] {$P_5$}; 
\addplot[color=blue,mark=o,mark size=1.2pt] coordinates {(1.6,2.03)}; 
\addplot[color=blue,mark=asterisk,mark size=1.2pt] coordinates {(1.6,2.14)}; 
\addplot[color=blue,mark=*,mark size=1.0pt] coordinates {(1.8,1.02)}; 
\addplot[color=blue,mark=x,mark size=1.2pt] coordinates {(1.8,1.02)}; 
\addplot[color=blue,mark=+,mark size=1.2pt] coordinates {(1.8,1.38)} node[above,xshift=-2mm] {$P_{10}$}; 
\addplot[color=blue,mark=o,mark size=1.2pt] coordinates {(1.8,1.17)}; 
\addplot[color=blue,mark=asterisk,mark size=1.2pt] coordinates {(1.8,1.22)}; 
\addplot[color=blue,mark=*,mark size=1.0pt] coordinates {(1.95,0.25)}; 
\addplot[color=blue,mark=x,mark size=1.2pt] coordinates {(1.95,0.25)}; 
\addplot[color=blue,mark=+,mark size=1.2pt] coordinates {(1.95,0.41)} node[left] {$P_{40}$}; 
\addplot[color=blue,mark=o,mark size=1.2pt] coordinates {(1.95,0.33)}; 
\addplot[color=blue,mark=asterisk,mark size=1.2pt] coordinates {(1.95,0.34)}; 
\addplot[color=teal,mark=*,mark size=1.0pt] coordinates {(2,2.24)}; 
\addplot[color=teal,mark=x,mark size=1.2pt] coordinates {(2,2.24)}; 
\addplot[color=teal,mark=+,mark size=1.2pt] coordinates {(2,2.35)} node[above] {$C_5$}; 
\addplot[color=teal,mark=o,mark size=1.2pt] coordinates {(2,2.36)}; 
\addplot[color=teal,mark=asterisk,mark size=1.2pt] coordinates {(2,2.28)}; 
\addplot[color=teal,mark=*,mark size=1.0pt] coordinates {(2,1.15)}; 
\addplot[color=teal,mark=x,mark size=1.2pt] coordinates {(2,1.15)}; 
\addplot[color=teal,mark=+,mark size=1.2pt] coordinates {(2,1.39)} node[above,xshift=2mm] {$C_{10}$}; 
\addplot[color=teal,mark=o,mark size=1.2pt] coordinates {(2,1.39)}; 
\addplot[color=teal,mark=asterisk,mark size=1.2pt] coordinates {(2,1.25)}; 
\addplot[color=teal,mark=*,mark size=1.0pt] coordinates {(2,0.29)}; 
\addplot[color=teal,mark=x,mark size=1.2pt] coordinates {(2,0.29)}; 
\addplot[color=teal,mark=+,mark size=1.2pt] coordinates {(2,0.39)} node[above,xshift=2mm] {$C_{40}$}; 
\addplot[color=teal,mark=o,mark size=1.2pt] coordinates {(2,0.39)}; 
\addplot[color=teal,mark=asterisk,mark size=1.2pt] coordinates {(2,0.35)}; 
\addplot[color=cyan,mark=*,mark size=1.0pt] coordinates {(3,1.68)}; 
\addplot[color=cyan,mark=x,mark size=1.2pt] coordinates {(3,1.68)}; 
\addplot[color=cyan,mark=+,mark size=1.2pt] coordinates {(3,1.96)} node[above] {$Q_3$}; 
\addplot[color=cyan,mark=o,mark size=1.2pt] coordinates {(3,1.96)}; 
\addplot[color=cyan,mark=asterisk,mark size=1.2pt] coordinates {(3,1.77)}; 
\addplot[color=cyan,mark=*,mark size=1.0pt] coordinates {(4,0.99)}; 
\addplot[color=cyan,mark=x,mark size=1.2pt] coordinates {(4,0.99)}; 
\addplot[color=cyan,mark=+,mark size=1.2pt] coordinates {(4,1.42)} node[above] {$Q_4$}; 
\addplot[color=cyan,mark=o,mark size=1.2pt] coordinates {(4,1.42)}; 
\addplot[color=cyan,mark=asterisk,mark size=1.2pt] coordinates {(4,1.08)}; 
\addplot[color=cyan,mark=*,mark size=1.0pt] coordinates {(5,0.56)}; 
\addplot[color=cyan,mark=x,mark size=1.2pt] coordinates {(5,0.56)}; 
\addplot[color=cyan,mark=+,mark size=1.2pt] coordinates {(5,1.05)} node[above,xshift=-2mm] {$Q_5$}; 
\addplot[color=cyan,mark=o,mark size=1.2pt] coordinates {(5,1.05)}; 
\addplot[color=cyan,mark=asterisk,mark size=1.2pt] coordinates {(5,0.65)}; 
\addplot[color=magenta,mark=*,mark size=1.0pt] coordinates {(3,1.43)} node[below] {${\rm Kg}_{5,2}$}; 
\addplot[color=magenta,mark=x,mark size=1.2pt] coordinates {(3,1.45)}; 
\addplot[color=magenta,mark=+,mark size=1.2pt] coordinates {(3,1.61)}; 
\addplot[color=magenta,mark=o,mark size=1.2pt] coordinates {(3,1.63)}; 
\addplot[color=magenta,mark=asterisk,mark size=1.2pt] coordinates {(3,1.52)}; 
\addplot[color=magenta,mark=*,mark size=1.0pt] coordinates {(4,0.5)} node[below,xshift=4mm] {${\rm Kg}_{7,3}$}; 
\addplot[color=magenta,mark=x,mark size=1.2pt] coordinates {(4,0.5)}; 
\addplot[color=magenta,mark=+,mark size=1.2pt] coordinates {(4,0.66)}; 
\addplot[color=magenta,mark=o,mark size=1.2pt] coordinates {(4,0.68)}; 
\addplot[color=magenta,mark=asterisk,mark size=1.2pt] coordinates {(4,0.58)}; 
\addplot[color=brown,mark=*,mark size=1.0pt] coordinates {(3.8,0.6)} node[below,xshift=-4mm,yshift=1mm] {$G_{20,0.2}$}; 
\addplot[color=brown,mark=x,mark size=1.2pt] coordinates {(3.8,1.22)}; 
\addplot[color=brown,mark=+,mark size=1.2pt] coordinates {(3.8,1.5)}; 
\addplot[color=brown,mark=o,mark size=1.2pt] coordinates {(3.8,0.68)}; 
\addplot[color=brown,mark=asterisk,mark size=1.2pt] coordinates {(3.8,0.83)}; 
\end{axis}
\end{tikzpicture}
&
\begin{tikzpicture}
\begin{axis}[
    xlabel={$2m/n$=average degree},
    ylabel={$\frac{D}{n}$},
    ylabel style={rotate=-90,at={(0.05,0.9)}},
    xmin=5, xmax=115,
    ymin=0, ymax=13,
    xtick={10,20,30,40,50,60,70,80,90,100,110},
    ytick={0,2,4,6,8,10,12},
    ymajorgrids=true,
    xmajorgrids=true,
    grid style=dashed,
    width=10cm,
    height=6cm,
]
\addplot[color=magenta,mark=*,mark size=1.0pt] coordinates {(10,0.55)}; 
\addplot[color=magenta,mark=x,mark size=1.2pt] coordinates {(10,0.56)}; 
\addplot[color=magenta,mark=+,mark size=1.2pt] coordinates {(10,0.76)} node[above,rotate=90,anchor=west,yshift=-1mm] {${\rm Kg}_{8,3}$}; 
\addplot[color=magenta,mark=o,mark size=1.2pt] coordinates {(10,0.94)}; 
\addplot[color=magenta,mark=asterisk,mark size=1.2pt] coordinates {(10,0.64)}; 
\addplot[color=blue,mark=*,mark size=1.0pt] coordinates {(7.83,1.9)}; 
\addplot[color=blue,mark=x,mark size=1.2pt] coordinates {(7.83,1.95)}; 
\addplot[color=blue,mark=+,mark size=1.2pt] coordinates {(7.83,3.13)} node[above,rotate=90,anchor=west] {$K_{3,4,5}$}; 
\addplot[color=blue,mark=o,mark size=1.2pt] coordinates {(7.83,2.11)}; 
\addplot[color=blue,mark=asterisk,mark size=1.2pt] coordinates {(7.83,2.11)}; 
\addplot[color=blue,mark=*,mark size=1.0pt] coordinates {(17.93,1.34)}; 
\addplot[color=blue,mark=x,mark size=1.2pt] coordinates {(17.93,1.37)}; 
\addplot[color=blue,mark=+,mark size=1.2pt] coordinates {(17.93,3.98)} node[above,rotate=90,anchor=west] {$K_{8,9,10}$}; 
\addplot[color=blue,mark=o,mark size=1.2pt] coordinates {(17.93,2.10)}; 
\addplot[color=blue,mark=asterisk,mark size=1.2pt] coordinates {(17.93,1.46)}; 
\addplot[color=blue,mark=*,mark size=1.0pt] coordinates {(28,1.11)}; 
\addplot[color=blue,mark=x,mark size=1.2pt] coordinates {(28,1.13)}; 
\addplot[color=blue,mark=+,mark size=1.2pt] coordinates {(28,4.89)} node[above,rotate=90,anchor=west] {$K_{13,14,15}$}; 
\addplot[color=blue,mark=o,mark size=1.2pt] coordinates {(28,2.24)}; 
\addplot[color=blue,mark=asterisk,mark size=1.2pt] coordinates {(28,1.28)}; 
\addplot[color=blue,mark=*,mark size=1.0pt] coordinates {(38,0.99)}; 
\addplot[color=blue,mark=x,mark size=1.2pt] coordinates {(38,1.00)}; 
\addplot[color=blue,mark=+,mark size=1.2pt] coordinates {(38,5.82)} node[above,rotate=90,anchor=west] {$K_{18,19,20}$}; 
\addplot[color=blue,mark=o,mark size=1.2pt] coordinates {(38,2.43)}; 
\addplot[color=blue,mark=asterisk,mark size=1.2pt] coordinates {(38,1.13)}; 
\addplot[color=blue,mark=*,mark size=1.0pt] coordinates {(48,0.92)}; 
\addplot[color=blue,mark=x,mark size=1.2pt] coordinates {(48,0.93)}; 
\addplot[color=blue,mark=+,mark size=1.2pt] coordinates {(48,6.76)} node[above,rotate=90,anchor=west,yshift=2mm] {$K_{23,24,25}$}; 
\addplot[color=blue,mark=o,mark size=1.2pt] coordinates {(48,2.64)}; 
\addplot[color=blue,mark=asterisk,mark size=1.2pt] coordinates {(48,1.02)}; 
\addplot[color=red,mark=*,mark size=1.0pt] coordinates {(26,1.54)}; 
\addplot[color=red,mark=x,mark size=1.2pt] coordinates {(26,1.57)}; 
\addplot[color=red,mark=+,mark size=1.2pt] coordinates {(26,3.16)} node[above,xshift=2mm,yshift=-1mm] {${\rm Ci}_{60,13}$}; 
\addplot[color=red,mark=o,mark size=1.2pt] coordinates {(26,3.21)}; 
\addplot[color=red,mark=asterisk,mark size=1.2pt] coordinates {(26,1.86)}; 
\addplot[color=brown,mark=*,mark size=1.0pt] coordinates {(19.6,1.35)} node[below,yshift=0.5mm] {$G_{50,0.5}$}; 
\addplot[color=brown,mark=x,mark size=1.2pt] coordinates {(19.6,2.53)}; 
\addplot[color=brown,mark=+,mark size=1.2pt] coordinates {(19.6,3.36)}; 
\addplot[color=brown,mark=o,mark size=1.2pt] coordinates {(19.6,1.95)}; 
\addplot[color=brown,mark=asterisk,mark size=1.2pt] coordinates {(19.6,2.23)}; 
\addplot[color=brown,mark=*,mark size=1.0pt] coordinates {(39.58,1.15)} node[below,yshift=1.5mm] {$G_{80,0.5}$}; 
\addplot[color=brown,mark=x,mark size=1.2pt] coordinates {(39.58,3.33)}; 
\addplot[color=brown,mark=+,mark size=1.2pt] coordinates {(39.58,5.69)}; 
\addplot[color=brown,mark=o,mark size=1.2pt] coordinates {(39.58,1.90)}; 
\addplot[color=brown,mark=asterisk,mark size=1.2pt] coordinates {(39.58,2.74)}; 
\addplot[color=brown,mark=*,mark size=1.0pt] coordinates {(59.92,1.98)} node[below,xshift=5mm] {$G_{120,0.5}$}; 
\addplot[color=brown,mark=x,mark size=1.2pt] coordinates {(59.92,4.09)}; 
\addplot[color=brown,mark=+,mark size=1.2pt] coordinates {(59.92,6.62)}; 
\addplot[color=brown,mark=o,mark size=1.2pt] coordinates {(59.92,2.76)}; 
\addplot[color=brown,mark=asterisk,mark size=1.2pt] coordinates {(59.92,3.46)}; 
\addplot[color=brown,mark=*,mark size=1.0pt] coordinates {(48.71,2.8)} node[below] {$G_{70,0.7}$}; 
\addplot[color=brown,mark=x,mark size=1.2pt] coordinates {(48.71,5.06)}; 
\addplot[color=brown,mark=+,mark size=1.2pt] coordinates {(48.71,7.66)}; 
\addplot[color=brown,mark=o,mark size=1.2pt] coordinates {(48.71,3.93)}; 
\addplot[color=brown,mark=asterisk,mark size=1.2pt] coordinates {(48.71,4.55)}; 
\addplot[color=brown,mark=*,mark size=1.0pt] coordinates {(70,3.28)} node[below,xshift=3mm] {$G_{100,0.7}$}; 
\addplot[color=brown,mark=x,mark size=1.2pt] coordinates {(70,6.67)}; 
\addplot[color=brown,mark=+,mark size=1.2pt] coordinates {(70,9.52)}; 
\addplot[color=brown,mark=o,mark size=1.2pt] coordinates {(70,4.29)}; 
\addplot[color=brown,mark=asterisk,mark size=1.2pt] coordinates {(70,5.59)}; 
\addplot[color=brown,mark=*,mark size=1.0pt] coordinates {(91.17,4.21)} node[below] {$G_{130,0.7}$}; 
\addplot[color=brown,mark=x,mark size=1.2pt] coordinates {(91.17,7.16)}; 
\addplot[color=brown,mark=+,mark size=1.2pt] coordinates {(91.17,11.17)}; 
\addplot[color=brown,mark=o,mark size=1.2pt] coordinates {(91.17,5.93)}; 
\addplot[color=brown,mark=asterisk,mark size=1.2pt] coordinates {(91.17,6.68)}; 
\addplot[color=brown,mark=*,mark size=1.0pt] coordinates {(112,4.47)} node[below,xshift=-4.5mm,yshift=1.5mm] {$G_{160,0.7}$}; 
\addplot[color=brown,mark=x,mark size=1.2pt] coordinates {(112,8.36)}; 
\addplot[color=brown,mark=+,mark size=1.2pt] coordinates {(112,12.49)}; 
\addplot[color=brown,mark=o,mark size=1.2pt] coordinates {(112,6.26)}; 
\addplot[color=brown,mark=asterisk,mark size=1.2pt] coordinates {(112,7.39)}; 
\end{axis}
\end{tikzpicture}
\\
\rotatebox[origin=l]{90}{\hspace{25mm}Algorithm~I\sss{}}
&
\begin{tikzpicture}
\begin{axis}[
    xlabel={$2m/n$=average degree},
    ylabel={$\frac{D}{n}$},
    ylabel style={rotate=-90,at={(0.05,0.9)}},
    xmin=0, xmax=5,
    ymin=0, ymax=3.5,
    xtick={0,0.5,1,1.5,2,2.5,3,3.5,4,4.5,5},
    ytick={0,0.5,1,1.5,2,2.5,3,3.5},
    ymajorgrids=true,
    xmajorgrids=true,
    grid style=dashed,
    width=8cm,
    height=6cm,
]
\addplot[color=red,mark=*,mark size=1.0pt]  coordinates {(0,1.9)} node[right] {$E_5$};
\addplot[color=red,mark=*,mark size=1.0pt]  coordinates {(0,0.91)} node[right] {$E_{10}$};
\addplot[color=red,mark=*,mark size=1.0pt]  coordinates {(0,0.36)} node[right] {$E_{25}$};
\addplot[color=blue,mark=*,mark size=1.0pt] coordinates {(1.6,2.68)}; 
\addplot[color=blue,mark=x,mark size=1.2pt] coordinates {(1.6,2.68)}; 
\addplot[color=blue,mark=+,mark size=1.2pt] coordinates {(1.6,2.86)} node[above] {$P_5$}; 
\addplot[color=blue,mark=o,mark size=1.2pt] coordinates {(1.6,2.62)}; 
\addplot[color=blue,mark=asterisk,mark size=1.2pt] coordinates {(1.6,2.76)}; 
\addplot[color=blue,mark=*,mark size=1.0pt] coordinates {(1.8,1.24)}; 
\addplot[color=blue,mark=x,mark size=1.2pt] coordinates {(1.8,1.24)}; 
\addplot[color=blue,mark=+,mark size=1.2pt] coordinates {(1.8,1.48)} node[above,xshift=-2mm] {$P_{10}$}; 
\addplot[color=blue,mark=o,mark size=1.2pt] coordinates {(1.8,1.30)}; 
\addplot[color=blue,mark=asterisk,mark size=1.2pt] coordinates {(1.8,1.37)}; 
\addplot[color=blue,mark=*,mark size=1.0pt] coordinates {(1.95,0.29)}; 
\addplot[color=blue,mark=x,mark size=1.2pt] coordinates {(1.95,0.29)}; 
\addplot[color=blue,mark=+,mark size=1.2pt] coordinates {(1.95,0.39)} node[left] {$P_{40}$}; 
\addplot[color=blue,mark=o,mark size=1.2pt] coordinates {(1.95,0.34)}; 
\addplot[color=blue,mark=asterisk,mark size=1.2pt] coordinates {(1.95,0.34)}; 
\addplot[color=teal,mark=*,mark size=1.0pt] coordinates {(2,2.98)}; 
\addplot[color=teal,mark=x,mark size=1.2pt] coordinates {(2,2.98)}; 
\addplot[color=teal,mark=+,mark size=1.2pt] coordinates {(2,2.96)} node[above] {$C_5$}; 
\addplot[color=teal,mark=o,mark size=1.2pt] coordinates {(2,3.00)}; 
\addplot[color=teal,mark=asterisk,mark size=1.2pt] coordinates {(2,2.97)}; 
\addplot[color=teal,mark=*,mark size=1.0pt] coordinates {(2,1.35)}; 
\addplot[color=teal,mark=x,mark size=1.2pt] coordinates {(2,1.35)}; 
\addplot[color=teal,mark=+,mark size=1.2pt] coordinates {(2,1.49)} node[above,xshift=2mm] {$C_{10}$}; 
\addplot[color=teal,mark=o,mark size=1.2pt] coordinates {(2,1.49)}; 
\addplot[color=teal,mark=asterisk,mark size=1.2pt] coordinates {(2,1.39)}; 
\addplot[color=teal,mark=*,mark size=1.0pt] coordinates {(2,0.32)}; 
\addplot[color=teal,mark=x,mark size=1.2pt] coordinates {(2,0.32)}; 
\addplot[color=teal,mark=+,mark size=1.2pt] coordinates {(2,0.38)} node[above,xshift=2mm] {$C_{40}$}; 
\addplot[color=teal,mark=o,mark size=1.2pt] coordinates {(2,0.38)}; 
\addplot[color=teal,mark=asterisk,mark size=1.2pt] coordinates {(2,0.34)}; 
\addplot[color=cyan,mark=*,mark size=1.0pt] coordinates {(3,1.96)}; 
\addplot[color=cyan,mark=x,mark size=1.2pt] coordinates {(3,1.96)}; 
\addplot[color=cyan,mark=+,mark size=1.2pt] coordinates {(3,2.20)} node[above] {$Q_3$}; 
\addplot[color=cyan,mark=o,mark size=1.2pt] coordinates {(3,2.20)}; 
\addplot[color=cyan,mark=asterisk,mark size=1.2pt] coordinates {(3,2.00)}; 
\addplot[color=cyan,mark=*,mark size=1.0pt] coordinates {(4,0.96)}; 
\addplot[color=cyan,mark=x,mark size=1.2pt] coordinates {(4,0.96)}; 
\addplot[color=cyan,mark=+,mark size=1.2pt] coordinates {(4,1.32)} node[above] {$Q_4$}; 
\addplot[color=cyan,mark=o,mark size=1.2pt] coordinates {(4,1.32)}; 
\addplot[color=cyan,mark=asterisk,mark size=1.2pt] coordinates {(4,1.00)}; 
\addplot[color=cyan,mark=*,mark size=1.0pt] coordinates {(5,0.50)}; 
\addplot[color=cyan,mark=x,mark size=1.2pt] coordinates {(5,0.50)}; 
\addplot[color=cyan,mark=+,mark size=1.2pt] coordinates {(5,0.87)} node[above,xshift=-2mm] {$Q_5$}; 
\addplot[color=cyan,mark=o,mark size=1.2pt] coordinates {(5,0.87)}; 
\addplot[color=cyan,mark=asterisk,mark size=1.2pt] coordinates {(5,0.52)}; 
\addplot[color=magenta,mark=*,mark size=1.0pt] coordinates {(3,1.58)} node[below] {${\rm Kg}_{5,2}$}; 
\addplot[color=magenta,mark=x,mark size=1.2pt] coordinates {(3,1.59)}; 
\addplot[color=magenta,mark=+,mark size=1.2pt] coordinates {(3,1.66)}; 
\addplot[color=magenta,mark=o,mark size=1.2pt] coordinates {(3,1.69)}; 
\addplot[color=magenta,mark=asterisk,mark size=1.2pt] coordinates {(3,1.63)}; 
\addplot[color=magenta,mark=*,mark size=1.0pt] coordinates {(4,0.45)} node[below,xshift=4mm] {${\rm Kg}_{7,3}$}; 
\addplot[color=magenta,mark=x,mark size=1.2pt] coordinates {(4,0.45)}; 
\addplot[color=magenta,mark=+,mark size=1.2pt] coordinates {(4,0.53)}; 
\addplot[color=magenta,mark=o,mark size=1.2pt] coordinates {(4,0.54)}; 
\addplot[color=magenta,mark=asterisk,mark size=1.2pt] coordinates {(4,0.48)}; 
\addplot[color=brown,mark=*,mark size=1.0pt] coordinates {(3.8,0.62)} node[below,xshift=-4mm,yshift=1mm] {$G_{20,0.2}$}; 
\addplot[color=brown,mark=x,mark size=1.2pt] coordinates {(3.8,1.03)}; 
\addplot[color=brown,mark=+,mark size=1.2pt] coordinates {(3.8,1.17)}; 
\addplot[color=brown,mark=o,mark size=1.2pt] coordinates {(3.8,0.64)}; 
\addplot[color=brown,mark=asterisk,mark size=1.2pt] coordinates {(3.8,0.75)}; 
\end{axis}
\end{tikzpicture}
&
\begin{tikzpicture}
\begin{axis}[
    legend style={font=\small},
    legend pos=north east,
    legend cell align={left},
    xlabel={$2m/n$=average degree},
    ylabel={$\frac{D}{n}$},
    ylabel style={rotate=-90,at={(0.05,0.9)}},
    xmin=5, xmax=115,
    ymin=0, ymax=13,
    xtick={10,20,30,40,50,60,70,80,90,100,110},
    ytick={0,2,4,6,8,10,12},
    ymajorgrids=true,
    xmajorgrids=true,
    grid style=dashed,
    width=10cm,
    height=6cm,
]
\addlegendimage{mark=*,black,only marks,mark size=1.5}
\addlegendimage{mark=x,black,only marks,mark size=1.5}
\addlegendimage{mark=+,black,only marks,mark size=1.5}
\addlegendimage{mark=o,black,only marks,mark size=1.5}
\addlegendimage{mark=asterisk,black,only marks,mark size=1.5}
\legend{\hspace{2mm}\texttt{min},\hspace{2mm}\texttt{revmin},\hspace{2mm}\texttt{max},\hspace{2mm}\texttt{revmax},\hspace{2mm}\texttt{rand}}
\addplot[color=magenta,mark=*,mark size=1.0pt] coordinates {(10,0.36)}; 
\addplot[color=magenta,mark=x,mark size=1.2pt] coordinates {(10,0.36)}; 
\addplot[color=magenta,mark=+,mark size=1.2pt] coordinates {(10,0.43)} node[above,rotate=90,anchor=west,yshift=-1mm] {${\rm Kg}_{8,3}$}; 
\addplot[color=magenta,mark=o,mark size=1.2pt] coordinates {(10,0.53)}; 
\addplot[color=magenta,mark=asterisk,mark size=1.2pt] coordinates {(10,0.38)}; 
\addplot[color=blue,mark=*,mark size=1.0pt] coordinates {(7.83,1.79)}; 
\addplot[color=blue,mark=x,mark size=1.2pt] coordinates {(7.83,1.84)}; 
\addplot[color=blue,mark=+,mark size=1.2pt] coordinates {(7.83,2.73)} node[above,rotate=90,anchor=west] {$K_{3,4,5}$}; 
\addplot[color=blue,mark=o,mark size=1.2pt] coordinates {(7.83,1.87)}; 
\addplot[color=blue,mark=asterisk,mark size=1.2pt] coordinates {(7.83,1.91)}; 
\addplot[color=blue,mark=*,mark size=1.0pt] coordinates {(17.93,0.76)}; 
\addplot[color=blue,mark=x,mark size=1.2pt] coordinates {(17.93,0.78)}; 
\addplot[color=blue,mark=+,mark size=1.2pt] coordinates {(17.93,2.34)} node[above,rotate=90,anchor=west] {$K_{8,9,10}$}; 
\addplot[color=blue,mark=o,mark size=1.2pt] coordinates {(17.93,1.27)}; 
\addplot[color=blue,mark=asterisk,mark size=1.2pt] coordinates {(17.93,0.86)}; 
\addplot[color=blue,mark=*,mark size=1.0pt] coordinates {(28,0.48)}; 
\addplot[color=blue,mark=x,mark size=1.2pt] coordinates {(28,0.49)}; 
\addplot[color=blue,mark=+,mark size=1.2pt] coordinates {(28,2.27)} node[above,rotate=90,anchor=west] {$K_{13,14,15}$}; 
\addplot[color=blue,mark=o,mark size=1.2pt] coordinates {(28,1.12)}; 
\addplot[color=blue,mark=asterisk,mark size=1.2pt] coordinates {(28,0.62)}; 
\addplot[color=blue,mark=*,mark size=1.0pt] coordinates {(38,0.35)}; 
\addplot[color=blue,mark=x,mark size=1.2pt] coordinates {(38,0.36)}; 
\addplot[color=blue,mark=+,mark size=1.2pt] coordinates {(38,2.23)} node[above,rotate=90,anchor=west] {$K_{18,19,20}$}; 
\addplot[color=blue,mark=o,mark size=1.2pt] coordinates {(38,1.05)}; 
\addplot[color=blue,mark=asterisk,mark size=1.2pt] coordinates {(38,0.46)}; 
\addplot[color=blue,mark=*,mark size=1.0pt] coordinates {(48,0.28)}; 
\addplot[color=blue,mark=x,mark size=1.2pt] coordinates {(48,0.29)}; 
\addplot[color=blue,mark=+,mark size=1.2pt] coordinates {(48,2.21)} node[above,rotate=90,anchor=west] {$K_{23,24,25}$}; 
\addplot[color=blue,mark=o,mark size=1.2pt] coordinates {(48,1.01)}; 
\addplot[color=blue,mark=asterisk,mark size=1.2pt] coordinates {(48,0.35)}; 
\addplot[color=red,mark=*,mark size=1.0pt] coordinates {(26,0.65)}; 
\addplot[color=red,mark=x,mark size=1.2pt] coordinates {(26,0.66)}; 
\addplot[color=red,mark=+,mark size=1.2pt] coordinates {(26,1.10)} node[above,xshift=4mm,yshift=-1mm] {${\rm Ci}_{60,13}$}; 
\addplot[color=red,mark=o,mark size=1.2pt] coordinates {(26,1.12)}; 
\addplot[color=red,mark=asterisk,mark size=1.2pt] coordinates {(26,0.7)}; 
\addplot[color=brown,mark=*,mark size=1.0pt] coordinates {(19.6,0.71)}; 
\addplot[color=brown,mark=x,mark size=1.2pt] coordinates {(19.6,1.19)}; 
\addplot[color=brown,mark=+,mark size=1.2pt] coordinates {(19.6,1.39)}  node[above,rotate=90,anchor=west,yshift=-2.5mm] {$G_{50,0.5}$}; 
\addplot[color=brown,mark=o,mark size=1.2pt] coordinates {(19.6,0.90)}; 
\addplot[color=brown,mark=asterisk,mark size=1.2pt] coordinates {(19.6,1.03)}; 
\addplot[color=brown,mark=*,mark size=1.0pt] coordinates {(39.58,0.48)}; 
\addplot[color=brown,mark=x,mark size=1.2pt] coordinates {(39.58,0.94)}; 
\addplot[color=brown,mark=+,mark size=1.2pt] coordinates {(39.58,1.34)} node[above,rotate=90,anchor=west,yshift=-2.5mm,xshift=-2mm] {$G_{80,0.5}$}; 
\addplot[color=brown,mark=o,mark size=1.2pt] coordinates {(39.58,0.56)}; 
\addplot[color=brown,mark=asterisk,mark size=1.2pt] coordinates {(39.58,0.76)}; 
\addplot[color=brown,mark=*,mark size=1.0pt] coordinates {(59.92,0.46)}; 
\addplot[color=brown,mark=x,mark size=1.2pt] coordinates {(59.92,0.84)}; 
\addplot[color=brown,mark=+,mark size=1.2pt] coordinates {(59.92,1.15)} node[above,rotate=90,anchor=west] {$G_{120,0.5}$}; 
\addplot[color=brown,mark=o,mark size=1.2pt] coordinates {(59.92,0.59)}; 
\addplot[color=brown,mark=asterisk,mark size=1.2pt] coordinates {(59.92,0.71)}; 
\addplot[color=brown,mark=*,mark size=1.0pt] coordinates {(48.71,0.9)}; 
\addplot[color=brown,mark=x,mark size=1.2pt] coordinates {(48.71,1.46)}; 
\addplot[color=brown,mark=+,mark size=1.2pt] coordinates {(48.71,1.77)} node[above,rotate=90,anchor=west,yshift=-3mm,xshift=-2mm] {$G_{70,0.7}$}; 
\addplot[color=brown,mark=o,mark size=1.2pt] coordinates {(48.71,1.07)}; 
\addplot[color=brown,mark=asterisk,mark size=1.2pt] coordinates {(48.71,1.26)}; 
\addplot[color=brown,mark=*,mark size=1.0pt] coordinates {(70,0.79)}; 
\addplot[color=brown,mark=x,mark size=1.2pt] coordinates {(70,1.41)}; 
\addplot[color=brown,mark=+,mark size=1.2pt] coordinates {(70,1.67)} node[above,rotate=90,anchor=west] {$G_{100,0.7}$}; 
\addplot[color=brown,mark=o,mark size=1.2pt] coordinates {(70,0.91)}; 
\addplot[color=brown,mark=asterisk,mark size=1.2pt] coordinates {(70,1.16)}; 
\addplot[color=brown,mark=*,mark size=1.0pt] coordinates {(91.17,0.82)}; 
\addplot[color=brown,mark=x,mark size=1.2pt] coordinates {(91.17,1.26)}; 
\addplot[color=brown,mark=+,mark size=1.2pt] coordinates {(91.17,1.57)} node[above,rotate=90,anchor=west] {$G_{130,0.7}$}; 
\addplot[color=brown,mark=o,mark size=1.2pt] coordinates {(91.17,0.98)}; 
\addplot[color=brown,mark=asterisk,mark size=1.2pt] coordinates {(91.17,1.12)}; 
\addplot[color=brown,mark=*,mark size=1.0pt] coordinates {(112,0.74)}; 
\addplot[color=brown,mark=x,mark size=1.2pt] coordinates {(112,1.20)}; 
\addplot[color=brown,mark=+,mark size=1.2pt] coordinates {(112,1.49)} node[above,rotate=90,anchor=west] {$G_{160,0.7}$}; 
\addplot[color=brown,mark=o,mark size=1.2pt] coordinates {(112,0.90)}; 
\addplot[color=brown,mark=asterisk,mark size=1.2pt] coordinates {(112,1.06)}; 
\end{axis}
\end{tikzpicture}
\end{tabular}
}
\caption{Experimental comparison of Algorithm~I (top row) and Algorithm~I\sss{} (bottom row) on a variety of input graphs.
See~\eqref{eq:delay} for the definition of the quantity plotted on the vertical axes.
The legend at the bottom right applies to all four plots, showing the marker styles used for the five different vertex orderings.}
\label{fig:exp}
\end{figure}
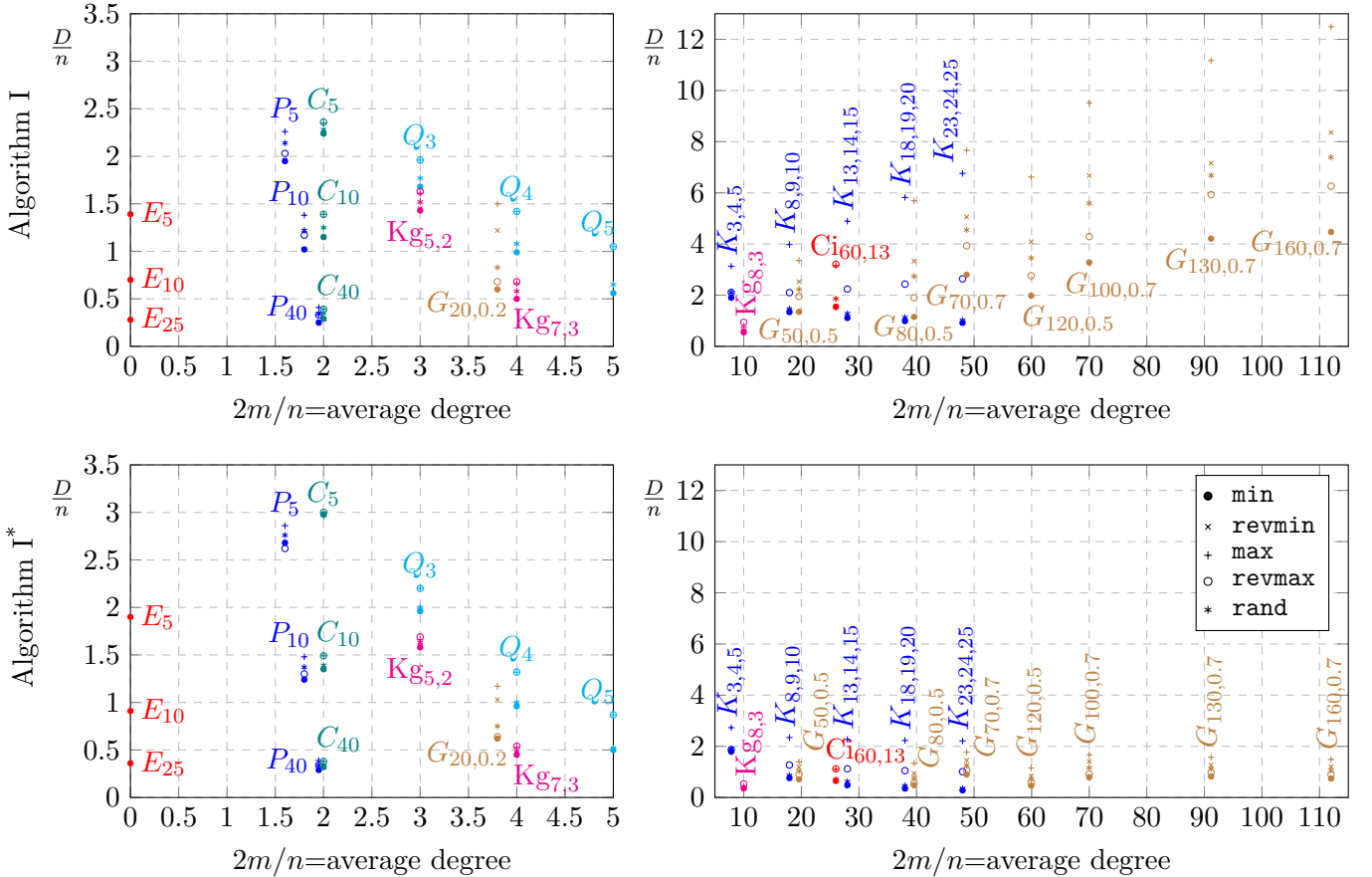

As can be seen for empty graphs with 5, 10 and 25 vertices, paths and cycles with 5, 10 and 40 vertices, and hypercubes of dimension 3, 4 and 5, the average delay per bit decreases considerably as the number of vertices increases.
To understand this, it helps to consider the family of empty graphs~$E_n$, and the correspondence to the binary reflected Gray code discussed in Section~\ref{sec:families} before.
In this case, the worst-case delay bound of Theorem~\ref{thm:algo-basic} and the average delay bound of Theorem~\ref{thm:algo-faster}, which are both linear in~$n$, are overly pessimistic, as the average delay is in fact constant.
Specificially, for $E_n$ the average delay is given by the average length of the prefixes that change in the binary reflected Gray code, which can be shown to be $2-o(1)<2$.
Consequently, when dividing by~$n$, we get a quantity that vanishes with~$n$.
More generally, for very sparse graphs that have a lot of independent sets, most changes occur on very short prefixes of~$x$, leading to a sublinear average delay for both Algorithm~I and Algorithm~I\sss{}, and thus the quantity $D/n$ decreases as $n$ grows.

The effect of vertex orderings can be seen most clearly for Algorithm~I when running it on the relatively dense complete tripartite graphs and random graphs.
Namely, the algorithm is fastest for the \texttt{min} ordering, as expected and predicted by Theorem~\ref{thm:algo-basic}.
The \texttt{revmax} ordering comes second, as it places large degree vertices at the end, which consequently leaves small degree vertices in the beginning (this is what matters most).
The \texttt{rand} ordering and \texttt{revmin} ordering come next.
The \texttt{max} ordering is clearly worst, as it places large degree vertices at the beginning.
The difference in running time between the best and worst ordering is substantial and grows with the average degree, emphasizing the necessity for choosing a good ordering.

A more fine-grained comparison of the \texttt{min} and \texttt{revmin} ordering for Algorithm~I\sss{} shows the following.
While Theorem~\ref{thm:algo-faster} suggests that the delay of Algorithm~I\sss{} depends primarily on the backward degeneracy~$d_<(G)$, which is minimized by the \texttt{revmin} ordering, experiments with more irregular graphs, such as~$K_{a,b,c}$ where $a<b<c$ are such that $b-a,c-b\gg 1$, show that actually the \texttt{min} ordering is by a factor of approximately~2 faster.
This somewhat counterintuitive behavior can be explained as follows:
While the forward and backward degeneracy minimize the maximum degree in either forward or backward direction, what is more important for the running time of the algorithms is that vertices with small labels have small degrees.
These are the vertices for which the neighborhoods are examined most frequently, typically exponentially more often than for vertices with large labels, as a consequence of the genlex ordering of the independent sets.
Therefore, while the \texttt{min} ordering places vertices with small forward degree first, this also seems to yield relatively small backward degrees in the beginning.
On the other hand, the \texttt{revmin} ordering places vertices with small backward degrees last, leaving vertices with potentially larger degrees in the beginning.
Therefore, in practice it might be best to always use the \texttt{min} ordering for both algorithms.

\bibliographystyle{alpha}
\bibliography{refs}

\end{document}